\documentclass[acmsmall,screen,nonacm]{acmart}

\AtBeginDocument{%
  }

\usepackage{boxedminipage}
\usepackage{bussproofs} 	
\usepackage{listings}	
\usepackage{qtree}
\usepackage{subcaption}

\newcommand{\OCAML}{\textsf{OCAML}}
\newcommand{\HASKELL}{\textsf{HASKELL}}
\newcommand{\EFF}{\textsf{EFF}}
\newcommand{\EPCF}{\textsf{EPCF}}
\newcommand{\HEPCF}{\textsf{HEPCF}}
\newcommand{\GEPCF}{\textsf{GEPCF}}
\newcommand{\PCF}{\textsf{PCF}}
\newcommand{\MSO}{\textsf{MSO}}
\newcommand{\ECPS}{\textsf{ECPS}}
\newcommand{\MSOp}[1]{\mathsf{MSO}_{#1}}
\newcommand{\la}{\lambda}
\newcommand{\laY}{\lambda Y}
\newcommand{\trees}{\mathbf{Tree}_{\Sigma}}
\newcommand{\treess}{\mathbf{Tree}_{\Sigma'}}
\newcommand{\etrees}[1]{\mathbf{Tree}_{\Sigma_{#1}}}
\newcommand{\set}[1]{\{ #1 \}}
\newcommand{\fix}{\mathtt{fix}}
\newcommand{\return}{\mathtt{return}}

\newcommand{\letin}[1]{\mathtt{let}~#1~\mathtt{in}}
\newcommand{\handle}[1]{\mathtt{with}~#1~\mathtt{handle}} 
\newcommand{\shandle}[1]{\mathtt{with}~#1~\mathtt{handle}^{\dagger}} 
\newcommand{\case}{\mathtt{case}}
\newcommand{\vect}[1]{\widetilde{#1}}
\newcommand{\type}[1]{\mathtt{#1}}
\newcommand{\ar}{\mathtt{ar}}
\newcommand{\ararrow}{\rightsquigarrow}
\newcommand{\sem}[1]{\llbracket #1 \rrbracket}

\newcommand{\toY}[1]{#1^{\star}}
\newcommand{\p}{\vdash}
\newcommand{\AXC}[1]{\AxiomC{#1}}
\newcommand{\UIC}[1]{\UnaryInfC{#1}}
\newcommand{\BIC}[1]{\BinaryInfC{#1}}
\newcommand{\TIC}[1]{\TrinaryInfC{#1}}
\newcommand{\DP}{\DisplayProof}
\newcommand{\vvskip}{\medskip}
\newcommand\midd{\; \mbox{\Large{$\mid$}}\;}
\newenvironment{framed}[0]{\begin{boxedminipage}{\linewidth}}{\end{boxedminipage}}
\newtheorem{definition}{Definition}
\newtheorem{theorem}{Theorem}
\newtheorem{lemma}{Lemma}
\newtheorem{example}{Example}
\newtheorem*{remark}{Remark}
\newtheorem{corollary}{Corollary}
\newcommand{\commentout}[1]{}
\lstdefinestyle{eff}{
	mathescape=true
	breaklines=true,
	basicstyle=\small,
	basewidth = {0.48em},
	emph={Flip,Set,Get,true,false,not,read,Read,raise,Raise,randomint,Randomint,randomfloat,Randomfloat,write,Write,print,Print,Choice,Do,Reward,Observe,Init,Op,choice,do,reward,observe,init,op}, emphstyle=\bfseries,
	emph={[2]unit,float,int,loc,bool},emphstyle={[2]\color{cyan!80!black}},
	keywordstyle=\color{blue!70!black},
	stringstyle=\color{orange!50!black},
	showstringspaces=false,
	comment=[s]{(*}{*)},
	commentstyle=\color{green!50!black},
	morestring=[b]",
	morekeywords={case,fix,perform,ref,begin,end,finally,continue,perform,type,use,rec,effect,handle,with,handler,let,return,fun,while,match,if,then,else,probability,in},
	literate= 
	{|->}{{$\mapsto$}}2
	{bot}{{$\bot$}}1
	{xor}{{$\oplus$}}1
	{lambda}{{$\lambda$}}1
	{odod}{{\textcolor{blue}{do}}}2
}

\begin{document}

	\title{On Model-Checking Higher-Order Effectful Programs}

\author{Ugo Dal Lago}
\affiliation{%
	\institution{University of Bologna, INRIA}
	\country{Italy}}
\authornote{Both authors contributed equally to this research.}
\email{ugo.dallago@unibo.it}
\orcid{1234-5678-9012}

\author{Alexis Ghyselen}
\affiliation{%
  \institution{University of Bologna}
  \country{Italy}}
\email{alexis.ghyselen@unibo.it}

\setcopyright{none}
\settopmatter{printacmref=false}

\renewcommand{\shortauthors}{Dal Lago and Ghyselen}

\begin{abstract}
  Model-checking is one of the most powerful techniques for verifying systems 
  and programs, which since the pioneering results by Knapik et al., Ong, and 
  Kobayashi, is known to be applicable to functional programs with higher-order 
  types against properties expressed by formulas of monadic second-order logic. 
  What happens when the program in question, in addition to higher-order 
  functions, also exhibits algebraic effects such as probabilistic choice or 
  global store? The results in the literature range from those, mostly positive, 
  about nondeterministic effects, to those about probabilistic effects, in the 
  presence of which even mere reachability becomes undecidable. This work takes a 
  fresh and general look at the problem, first of all showing that there is an 
  elegant and natural way of viewing higher-order programs producing algebraic 
  effects as ordinary higher-order recursion schemes. We then move on to 
  consider effect handlers, showing that in their presence the model checking 
  problem is bound to be undecidable in the general case, while it stays 
  decidable when handlers have a simple syntactic form, still sufficient to capture so-called \emph{generic effects}. Along the way we hint at how a general specification language could look like, this way justifying some of the results in the literature, and deriving new ones.
\end{abstract}

%\begin{CCSXML}
%<ccs2012>
%  <concept_id>10010520.10010553.10010562</concept_id>
%  <concept_desc>Computer systems organization~Embedded systems</concept_desc>
%  <concept_significance>500</concept_significance>
%  <concept_id>10010520.10010575.10010755</concept_id>
%  <concept_desc>Computer systems organization~Redundancy</concept_desc>
%  <concept_significance>300</concept_significance>
% <concept_id>10010520.10010553.10010554</concept_id>
%  <concept_desc>Computer systems organization~Robotics</concept_desc>
%  <concept_significance>100</concept_significance>
%  <concept_id>10003033.10003083.10003095</concept_id>
%  <concept_desc>Networks~Network reliability</concept_desc>
%  <concept_significance>100</concept_significance>
%</ccs2012>
%\end{CCSXML}

%\ccsdesc[500]{Computer systems organization~Embedded systems}


\newif\ifLong

%\Longfalse
\Longtrue

\maketitle

\section{Introduction}
Verifying the correctness of programs endowed with higher-order functions is a 
very challenging problem, which can be addressed with various methodologies, 
from type systems~\cite{HughesPareto1996:SizedTypes,RowanPfenning2000:Intersection,FreemanPfenning1991:Refinement} to program logics~\cite{Jung2018:IRIS,Brady2013:Idris}, from 
symbolic execution~\cite{King1976:SymbolicExecution,Tobin2012:HOSymExec} to verified compilation~\cite{Leroy2009:CompCert}. 
An approach with some peculiarities is that of higher-order model checking (HOMC in the following), 
which consists in seeing the 
program at hand as a structure, then checking whether it renders 
a logical formula capturing the desired property true, namely whether it is a 
\emph{model} of it. Saying it another way, HOMC can be seen as the application
of the model checking paradigm \cite{Clarke1997:ModelChecking,Clarke2018:ModelCheckingBook} to higher-order programs. 
One of the characteristics of this approach is that, contrary to most others, it is often both sound \emph{and 
complete}. As a consequence, the kind of languages to which the methodology can 
be applied are very often \emph{not} Turing-complete, the underlying 
verification problem being undecidable even for very simple logics.

The feasibility of higher-order model checking was scrutinized in the early 2000s, the objective 
being to extend classic results about model checking \emph{recursion 
schemes}~\cite{Courcelle1995:MSOTrees} to higher-order generalizations of the latter. The results obtained 
were initially very interesting but partial~\cite{Knapi2002:HORSareeasy,Knapik2001:MSOforHORS}, only concerning certain restricted forms of 
higher-order recursion schemes. The quest came to an end in 2006 with Ong's groundbreaking
result~\cite{Ong2006:HOMC} on the decidability of the model checking problem for trees generated by 
general higher-order recursion schemes against formulas of MSO (or, equivalently, 
of formulas the $\mu$-calculus or alternating parity tree automata \cite{Emerson91:TreeAutomata,Gradel2003:Automata}). 
This result was followed by many other ones \cite{Kobayashi2009:TypesHOMC,KobayashiOng2009:TypeSystemHOMC,Hague2008:CPDAandHORS,Salvati2011:KrivineMachineHORS,Broadbemt2010:HORSandLogicalReflection,Carayol2012:CPDAandHORS,Walukiewicz2016:AutomataHOMC}, whose goal was that of 
understanding the deep computational nature of the problem, at the same time 
generalizing the decidability result and building
concrete verification tools \cite{Kobayashi2011:LinearAlgorithm,Naetherway2012:TraversalAlgo,BroadbentKobayashi2013:AlgoSaturation,Ramsay2014:RefinementAlgoHOMC}, readers can refer to \cite{Ong2015:HOMC} for an overview. 

Among the various extensions of higher-order schemes considered in the 
literature, we should certainly mention extensions aimed at capturing 
higher-order recursion schemes subject to more permissive type disciplines than 
that of simple types, namely the one to which Ong's classic result applies. 
As an example, higher-order recursion schemes with 
recursive types have been recently considered~\cite{Kobayashi2013:HOMCRecursiveTypes}.
We should also mention some attempts at making the technique applicable to 
programs which are not pure, but which can produce, for example, 
nondeterministic and probabilistic effects. In the second case the decidability 
results scale back \cite{Kobayashi2020:Termination}, with undecidability showing up already at 
order three and for mere reachability properties. In the first case, instead 
results remain essentially unchanged \cite{Tsukada2014:NonDeterministic}. 

The motivation from which this work originates is precisely that of understanding 
the deep reasons for the aforementioned discrepancy, at the same time giving a 
general account of the HOMC problem in presence of effects. In doing so, we will
consider effects as being captured by algebraic operations~\cite{PlotkinPower2003:AlgebraicOperations}, 
the latter producing some pre-defined effects or interpreted by way of effect handlers~\cite{Plotkin2009:Handlers,KammarICFP2013:Handlers,Hillerstorm2017:Continuation}. In other
words, we will consider well-established ways of capturing effects in higher-order
$\lambda$-calculi. On the side of specifications, we analogously try not to consider ad-hoc formalisms, and look for conservative extensions of MSO in which the properties of interest can be captured in a unifying way. In particular, in the algebraic approach to effects, trees are always considered up to an equational theory, describing how the different algebraic operation should behave. In our approach, this means that it is particularly important to be able to define specifications that take into account this equational theory, and this is captured in our approach by the notion of an $\emph{observation}$ \cite{JohannSimpson2010:AlgebraicEffects,Simpson2019:BehaviouralEquivalenceEffects,MatacheStaton2019:LogicAlgebraicEffects}.  

The contributions of this paper are threefold:
\begin{itemize}
\item
We first of all consider a finitary version of Simpson and Voorneveld's \EPCF, a calculus with 
full recursion and algebraic effects, showing that the computation tree semantics of any
\EPCF\ computation $C$ is precisely the one of a $\laY$-term $C^*$ obtained by CPS-translating 
$C$. Since the $\laY$-calculus is well-known to be equiexpressive to higher-order recursion schemes \cite{NakamuraetAl2020:AverageCaseHardness}, we obtain that computation trees generated by \EPCF\ can be automatically checked against MSO specifications. This is in Section~\ref{s:EPCF} and Section~\ref{sect:translationEPCFtoY}.
\item
Then, we turn our attention to more general and more expressive specifications. Among the many proposals for logic for algebraic effects, we consider a variation on the one proposed by Simpson and Voorneveld, in which any effectful computations can be tested through the notion of an observation. We then prove that only certain notions of observations give rise to decidable model-checking problems, this way justifying some of the existing results in the literature, at the same time proving new ones. This is in Section~\ref{s:logic}.
\item
Finally, we consider the impact of effect handlers to the HOMC problem. We show that handlers are indeed harmful to decidability, at least when general, well-established notions of handlers are considered, including shallow and deep handlers. We conclude by considering a rather restricted class of handlers which are sufficiently expressive to capture generic effects \cite{PlotkinPower2003:AlgebraicOperations} but for which model checking remains decidable. This is in Section~\ref{s:handlers} and Section~\ref{s:modelcheckinghandlers}.
\end{itemize}

All in all, the results above provide a rather clear picture about how far one can go in applying existing HOMC methodologies to effectful programs. The take-home messages are that in principle, such techniques can be reused, provided the underlying notion of observation does not give rise to too complex specifications, while handlers are potentially very dangerous, and should be used with great care. The technical core of the paper is in Section~\ref{s:preliminaries} to Section~\ref{s:modelcheckinghandlers}, while Section~\ref{s:informal} serves as a gentle introduction to higher-order programming with effects and its verification. Related work is discussed in Section~\ref{s:relatedwork}.

\section{Higher-Order Effectful Programs, and how to Model-Check Them}
\label{s:informal}
While the $\lambda$-calculus is the reference paradigmatic model for pure 
functional programming, a standard way of raising \emph{effects} from within functional 
programs consists in invoking algebraic operations \cite{PlotkinPower2003:AlgebraicOperations}, 
each of them corresponding to a particular way of producing an observable effect. Even when 
the underlying programming language does not offer algebraic operations 
natively, many impure constructs can be interpreted this way. Consider, as an 
example, the \OCAML\ program in Figure~\ref{fig:ocamlp}, call it $\mathtt{P}$, 
which manipulates two ground global variables $\mathrm{r}$ and $\mathrm{q}$ 
through a recursive higher-order function $\mathrm{f}$.
\begin{figure}
	\centering
	\begin{subfigure}{0.55\textwidth}
		%\centering
\begin{lstlisting}[style = eff]
let r = ref true;;
let q = ref true;;
let rec f g =
  let y = !r in 
  let z = !q in 
  if (g y z) then failwith("Failure") else begin
    r := (not z);
    q := (not y);
    f g;
  end
in f (fun x y |-> x <> y)
\end{lstlisting}
		\caption{An \OCAML\ program $\mathtt{P}$ which never fails.}
		\label{fig:ocamlp}
	\end{subfigure}
\hspace{.05\textwidth}
	\begin{subfigure}{0.35\textwidth}
		%\centering
\begin{lstlisting}[style = eff]
Set((r,true);lambda_. 
Set((q,true);lambda_. 
let F = return(fix f.lambdag.
  Get(r,lambday. 
  Get(q,lambdaz.
  let x = (g y z) in 
  case(x, Raise(), 
    Set((r,not z);lambda_. 
    Set((q,not y);lambda_.
    f g ))
  )))) 
in F (lambda(x,y). x xor y)
))
\end{lstlisting}
		\caption{An \EPCF\ term $M_{\mathtt{P}}$.}
		\label{fig:LYP}
	\end{subfigure}	
%	\begin{subfigure}{0.3\textwidth}
%		\centering
%		\begin{lstlisting}[style = eff]
%			effect Set : loc * bool -> unit
%			effect Get : loc -> bool
%			effect Raise : unit -> 0
%			perform (Set r true);
%			perform (Set q true);
%			let rec f g =
%			let y = perform (Get r) 
%			let z = perform (Get q)
%			if (g y z) then perform (Raise()) else begin
%			perform (Set r (not z));
%			perform (Set q (not y));
%			f g;
%			end
%			in f (fun x y |-> x <> y)
%		\end{lstlisting}
%	\end{subfigure}
\caption{}
\Description{}
\end{figure}
%\begin{lstlisting}[style = eff]
%let r = ref true;;
%let q = ref true;;
%let rec f g =
%  let y = !r in 
%  let z = !q in 
%  if (g y z) then failwith("Failure") else begin
%    r := (not z);
%    q := (not y);
%    f g;
%  end
%in f (fun x y |-> x <> y)
%\end{lstlisting}
As can be easily realized, whenever the conditional is executed the two 
references $\mathrm{r}$ and $\mathrm{q}$ contain \emph{the same} boolean value. 
As a consequence, the Failure exception is never raised, and the program can be 
considered safe. Could we automatically verify the latter by way of 
higher-order model checking? Let us try to see if this is possible. 

We can assume $\type{Loc}$ to be a type of locations inhabited by 
$\mathrm{q}$ and $\mathrm{r}$ only, and that the program can invoke any effect-raising 
operations from the following typed signature: 
$$
\Sigma = \set{\mathbf{Get} : \type{Loc} \ararrow 2, \mathbf{Set} : \type{Loc} 
\times \type{Bool} \ararrow 1, \mathbf{Raise} : \type{Unit} \ararrow 0}. 
$$
%We allow some syntactic sugar for the sake of clarity, for example $g~y~z 
%%%\triangleq \letin{a = g~y}~a~z$ and $\lambda(x,y). C \triangleq \lambda x. 
%%%\return(\la 
%%%%y. C)$. 
Some standard abbreviations allow us to form the term in Figure~\ref{fig:LYP}, 
call it $M_{\mathtt{P}}$, whose structure is very similar to the one of 
$\mathtt{P}$. In doing so, we have adopted a syntax close to Simpson and 
Voorneveld's \EPCF.
%\begin{lstlisting}[style = eff]
%Set((q,true);lambda_. 
%Set((r,true);lambda_. 
%let F = return(fix f.lambdag.
%  Get(q,lambday. 
%  Get(r,lambdaz.
%  let x = (g y z) in 
%  case(x, Raise(), 
%    Set((q,not z);lambda_. 
%    Set((r,not y);lambda_.
%    f g ))
%  )))) 
%in F (lambda(x,y). x xor y)
%))
%\end{lstlisting}
Observe how \OCAML's reference commands have become algebraic operations from $\Sigma$, 
and how $\mathbf{Get}$ has arity equal to two, accounting for the fact that  
the program can proceed depending on the value read from memory. By the way, 
$M_{\mathtt{P}}$ closely corresponds to the way one would write $\mathtt{P}$ in 
languages like \EFF\ \cite{Pretnar2015:EFF}.
%\begin{lstlisting}[style = eff]
%effect Set : loc * bool -> unit
%effect Get : loc -> bool
%effect Raise : unit -> 0
%perform (Set r true);
%perform (Set q true);
%let rec f g =
%  let y = perform (Get r) 
%  let z = perform (Get q)
%  if (g y z) then perform (Raise()) else begin
%    perform (Set r (not z));
%    perform (Set q (not y));
%    f g;
%  end
%in f (fun x y |-> x <> y)
%\end{lstlisting}
\begin{figure}
	\centering
	\begin{subfigure}{0.4\textwidth}
	\centering
	\includegraphics[scale=0.7]{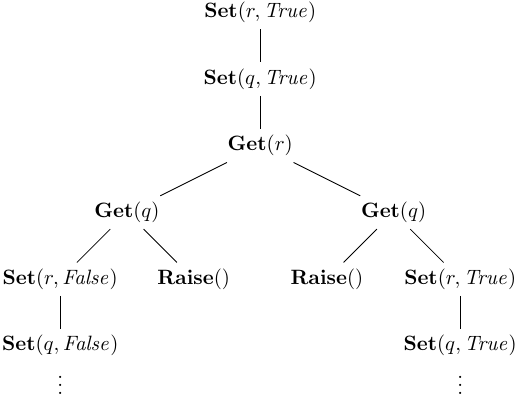}
	\caption{Tree Generated by $M_{\mathtt{P}}$.}
	\label{fig:storetree}
	\end{subfigure}
	\hspace{0.1\textwidth}
	\begin{subfigure}{0.4\textwidth}
	\centering
	\includegraphics[scale=0.7]{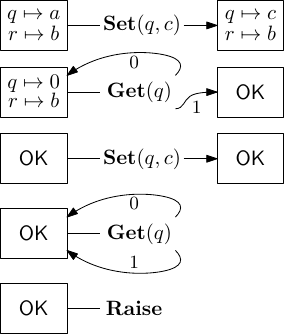}
	\caption{Automaton Capturing the safety of a program working with the 
	boolean variables $\mathrm{r}$ and $\mathrm{Q}$.}
	\label{fig:storeautomaton}
	\end{subfigure}	
	\caption{}
	\Description{}
\end{figure}

But why do algebraic operation could help in the task of verifying the safety 
of $\mathtt{P}$? Actually, the evaluation of programs which invoke algebraic 
operations naturally gives rise to a so-called \emph{effect tree}. The 
effect tree produced in output by $M_\mathtt{P}$ looks like one in 
Figure~\ref{fig:storetree}. Through it, one can verify that $\mathbf{Raise}$ is 
never executed by exhaustively considering all branches of the tree, and 
verifying that those which are somehow coherent with the store operations do 
not end in a leaf labeled with $\mathbf{Raise}$. Here, coherent branches are 
those that, e.g., when encountering $\mathbf{Get}(q)$ proceed left 
(respectively, right) depending on the last $\mathbf{Set}(q,a)$ operation 
performed. This reasoning can indeed be encoded in an alternating parity tree automata, and thus by 
a \MSO\ or $\mu$-calculus formula \cite{Gradel2003:Automata,Emerson91:TreeAutomata}. 
Let us briefly describe how this automaton can be constructed (here, by the way, we only need a top-down 
deterministic automata). Its set of states comprises all assignments of boolean values 
to the variables $\mathrm{r}$ and $\mathrm{q}$, together with a special state 
$\mathsf{OK}$. Intuitively, the latter captures incoherent states, and any 
action is accepted from there. Some example transitions are those in 
Figure~\ref{fig:storeautomaton}. The first two capture the expected behavior 
of $\mathbf{Set}$ and $\mathbf{Get}$ on stores, while the last three are there 
to model the fact that every action, including $\mathbf{Raise}$ is allowed in 
the state $\mathsf{OK}$. Of course, the latter is \emph{not} available in 
ordinary stores instead.
%Let us For the sake of simplicity, we consider here a simple deterministic 
%%%automaton that goes from the root of the tree to the leaves. We have the 
%%%%states $Q = \set{q_{q = 0, r = 0},q_{q = 0, r = 1},q_{q = 1, r = 0},q_{q = 
%%%%1, 
%%%%%r = 1},q_{ok}}$, one state for each possible configuration of the memory 
%%%%%and a 
%%%%%%that accepts everything. A run of this automata is accepted if and only 
%%%%%%if we 
%%%%%%%have a valid run of the tree, and the only way to not have an available 
%%%%%%%%transition is to see a $Raise$ with a state different from $q_ok$. 
%\begin{align*}
%	\delta(q_{l = x, r = y},Set(r,0)) &= q_{q = x, r = 0}  &\delta(q_{q = x, r 
%%%= 0},Get(r)) &=  q_{q = x, r = 0} \cdot q_{ok} \\
%	\delta(q_{ok},Set(\ell,b)) &= q_{ok} &\delta(q_{ok},Get(\ell)) &= q_{ok} 
%%%\cdot q_{ok} \\
%    \delta(q_{ok},Raise) &= \epsilon && 
%\end{align*}

All in all, then, we indeed have a way to turn impure functional programs into 
terms of a calculus called \EPCF\ which generate trees which are, at least 
superficially, amenable to model checking. There is still a missing link 
though: it could well be that \EPCF\ is simply too expressive, and that the 
model-checking problem is undecidable. Fortunately, however, we can turn 
\EPCF\ programs into $\laY$-terms, for which \MSO\ model checking is indeed decidable \cite{NakamuraetAl2020:AverageCaseHardness}. 
This is precisely what we prove in Section 
\ref{sect:translationEPCFtoY} below. 

Summing up, HOMC indeed seems to be helpful when considered on effect trees, because 
the program we are interested in verifying, and arguably any term working with 
global references over finite domains can be turned into a term in the $\lambda$Y-calculus, while 
the property becomes an \MSO\ formula. But how about other effects? Can 
we turn the construction above into something more general and systematic? 
While specific kinds of effects are considered in the literature~\cite{Kobayashi2009:TypesHOMC,Kobayashi2020:Termination,Ong2015:HOMC}, no 
general result is known. In fact, each effect comes equipped with its 
intended notion of observation \cite{JohannSimpson2010:AlgebraicEffects,MatacheStaton2019:LogicAlgebraicEffects,Simpson2019:BehaviouralEquivalenceEffects,DalLago2017:EffectfulBisim}. Which ones of those are simple enough to 
guarantee that model checking useful properties stays decidable? 

For the sake of convincing the reader that the questions above are not trivial, 
let us consider another example of an effectful higher-order program, borrowed from \cite{Kobayashi2020:Termination}
\begin{lstlisting}[style = eff]
let F = return(fix f. lambda_. 
  Flip(lambda b.
    case(b, return(),
      let x = f () in let y = f() in f()
))) in F () 
\end{lstlisting}
The algebraic operation $\mathbf{Flip}(\la b. C)$ should be understood as flipping a random coin, storing the result in $b$, and continuing as $C$. Intuitively, the program described above creates a procedure $f$ that throws a coin. If the coin returns head, the procedure stop. Otherwise, the procedure is executed three times. This program generates an effect tree $T_0$ where for every $n$, the tree $T_n$ can be defined as:
\[
T_0 \triangleq 
\Tree[.$\mathbf{Flip()}$ 
[.$\return()$ ] 
[.$T_2$ ]
]
\qquad 
\qquad 
\qquad 
T_{n+1} \triangleq 
\Tree[.$\mathbf{Flip()}$ 
[.$T_n$ ] 
[.$T_{n + 3}$ ]
]
\]
Indeed, the sequential composition in the third argument of the $\case$ operator gives raise to a stack of continuations, meaning that a leaf $\return()$ is replaced by the tree computed by the continuation. Here, the tree $T_n$ should be understood as a call to the recursive function $f$ with a continuation corresponding to $n$ calls to this function. This behavior is similar to the one of a random walk, but the program described above does not use any infinite type, which is essential for Higher-Order Model Checking. In fact, this tree can be computed by a term of the $\la$Y-calculus, using a CPS translation, see Example 2.10 from \cite{Kobayashi2020:Termination}: 
\begin{align*}
	(Y (\la F,k,x. \mathbf{Flip}~(k~x)~(F~(F~(F~k))~x)))~(\la x.x)~\return
\end{align*}
Here, the use of higher-order functions is important, since the function $F$ takes as a first input a continuation, and this continuation 
determines the number of calls to the procedure that still need to be executed. 
For such a program, it is natural to wonder whether $\return$ is called or not (if the program terminates), 
but given the presence of (probabilistic) nondeterminism, there are various ways in which 
this can be spelled out. Do we mean that the program must (or may) reach 
$\return$? Or do we rather mean that the program reach $\return$ with 
probability $1$? The latter question seems the most appropriate, given that 
probabilistic choice is captured by the subdistribution monad 
$\mathsf{D}(\cdot)$ and that $\mathsf{D}(\return)$ is just a real number between $0$ 
and $1$, i.e., the probability of not diverging. If this is the case, however, 
recent results by Kobayashi, Dal Lago, and Grellois \cite{Kobayashi2020:Termination} show 
that HOMC is \emph{not} decidable in general, so the construction of an \MSO\ formula (or an automaton)
like the one above is simply not possible. May or must termination, instead, can
easily be captured: may termination consists in exploring the tree and finding at least one $\return$, while must termination 
is satisfied if the tree has no infinite branch, and all leaves are $\return$ leaves, which is a property that can be encoded as a parity condition of an alternating parity automaton.  

To sum up, we can see effectful higher-order programs as effect trees, computed by a term of the $\la$Y-calculus (or equivalently, higher-order recursion schemes) for which we know that \MSO\ model-checking is decidable. This works particularly well for some effects, for example global store with finite domains, since the properties of interest can be expressed as an \MSO\ formula. However, for some more involved effects, such as probabilistic choice, some important properties are not expressible in \MSO. The aim of this work is to take a general look at this problem, and understand the deep reason behind this discrepancy. 

\section{Preliminaries about Higher-Order Model Checking}
\label{s:preliminaries}
\subsection{Infinite Trees Generated by $\la$Y-Terms}

In Higher-Order Model Checking, models are traditionally taken to be infinite 
trees produced by so-called Higher-Order Recursion Schemes (HORS in the 
following). In this work, we rather consider B\"ohm trees generated by ground 
type terms in the $\laY$-calculus with first-order constants, which is well 
known to express the same class of trees as HORSs \cite{NakamuraetAl2020:AverageCaseHardness}. 
This section is devoted to presenting some preliminaries about the calculus and the model checking 
problem for it.

Formally, the $\laY$-calculus can be seen as the simply-typed 
$\lambda$-calculus 
extended with full recursion and with first-order function symbols:
\begin{align*}
	&\mbox{(Types)} & T,U &::= o \midd T \rightarrow U \\ 
	&\mbox{(Terms)} & M,N &::= x \midd \la x.M \midd M~N \midd Y M \midd f \in 
	\Sigma
\end{align*}
The \emph{signature} $\Sigma$ is a set of first-order constants, such that each 
$f \in \Sigma$ comes equipped with an arity $\ar(f) \ge 0$ capturing the 
fact that the type of $f$ is $\underbrace{o \rightarrow o \rightarrow 
o}_{\ar(f)} \rightarrow o$, often abbreviated as $o^{\ar(f)} \rightarrow o$. By 
an abuse of notation, we usually write $f : \ar(f)$ to specify the arity of  
$f$ in a given signature when this does not cause ambiguity.
%and we denote by $o^{\ar(f)} \rightarrow o$ the type corresponding 
%to this chain of $\ar(f)$ arrows. In particular, the type $o^0 \rightarrow o$ 
%denotes the type $o$. 
Typing rules are standard, and can be found in Figure~\ref{f:lambdaYtypes}. 

\begin{figure}
	\begin{framed}
		\begin{center}
			\AXC{}
			\UIC{$\Gamma, x : T \p x : T$}
			\DP 
			\qquad 
			\AXC{$\Gamma, x: T \p M : U $}
			\UIC{$\Gamma \p \la x. M : T \rightarrow U$}
			\DP 
			\qquad 
			\AXC{$\Gamma \p M : T \rightarrow U$}
			\AXC{$\Gamma \p N : T$}
			\BIC{$\Gamma \p M~N : U$}
			\DP 
			\\
			\vvskip 
			\AXC{$\Gamma \p M : T \rightarrow T$}
			\UIC{$\Gamma \p Y M : T$}
			\DP 
			\qquad
			\AXC{} 
			\UIC{$\Gamma \p f : o^{\ar(f)} \rightarrow o$}
			\DP 
		\end{center}
	\end{framed}
	\caption{Typing Rules for $\la$Y-terms}
	\label{f:lambdaYtypes}
	\Description{}
\end{figure} 

We see the $\laY$-calculus as a tool to generate infinite trees. In order to 
precisely define the tree generated by a typable term, one has to define a form 
of dynamic semantics, which we here take as \emph{weak head reduction}: we 
never reduce the argument of any application, and we never evaluate terms in 
the scope of $\lambda$-abstractions. Rules are again standard, and can be found in 
Figure~\ref{f:lambdaYreduction}.
\begin{figure}
	\begin{framed}
		\begin{center}
			\AXC{}
			\UIC{$(\la x. M) N \rightarrow M[N / x]$}
			\DP 
			\qquad 
			\AXC{}
			\UIC{$Y M \rightarrow M (Y M)$}
			\DP 
			\qquad 
			\AXC{$M \rightarrow M' $}
			\UIC{$M~N \rightarrow M'~N$}
			\DP 
		\end{center}
	\end{framed}
	\caption{Weak Head Reduction Rules}
	\label{f:lambdaYreduction}
	\Description{}
\end{figure} 
It is relatively easy to see that typed closed terms in \emph{weak head normal 
form}, i.e. typed terms with no free variables that cannot be further reduced are 
precisely the (typable) terms generated by the following grammar:
$$ 
WHNF ::= \la x.M \midd f~M_1~M_2~\cdots~M_n 
$$ 
Indeed, a $\la$-abstraction is in normal form, constants are in normal form, 
and $M~N$ is in normal form if and only if $M$ is in normal 
form and $M$ is not a $\la$-abstraction. As a consequence, we can easily 
realize that terms of ground types in weak head normal form have the shape 
$f~M_1~M_2~\cdots~M_{\ar(f)}$, where each $M_i$ is itself a term of ground 
type. This naturally suggests a potentially infinite process turning any such 
term of ground type into a tree with root $f$ and $\ar(f)$ subtrees obtained by 
evaluating $M_1~\cdots~M_{\ar(f)}$, respectively. This can be made formal as 
follows:

\begin{definition}[Infinite Trees]
	The set of (potentially infinite) trees generated by a signature $\Sigma$,
	denoted $\trees$, is coinductively defined by the grammar:
	$$ t ::= f(t_1,\cdots,t_{\ar(f)}) \midd \bot $$
	where $f:\ar(f) \in \Sigma$.
 \end{definition}
The constant $\bot$ represents a non-terminating non-productive computation 
that does not generate any function symbol. As an example, if $\Sigma$ is the 
signature $\set{g : 2, f : 1, a : 0}$, one can form the non-regular infinite 
tree $g(a,g(f(a),g(f(f(a)),\cdots))$.
\begin{comment}
\[ 
\Tree[.$g$ 
  [.$a$ ] 
  [.$g$ 
    [.$f$
      [.$a$ ]
    ]
    [.$g$ 
      [.$\cdots$ ]
      [.$\cdots$ ]
    ] 
  ]
]
\]
\end{comment}
 We are now ready to define how any closed 
ground $\la$Y-term generates such an infinite tree:
\begin{definition}[Böhm Trees of Closed Ground Terms]
	Given a closed term $M$ of type $o$, i.e. we have $\cdot \p M : o$,  the \emph{Böhm tree} of $M$, 
	denoted $BT(M)$, is defined by way of the following 
	essentially infinitary process. Starting from $M$, we apply $\rightarrow$ 
	\emph{ad infinitum}. This can have two possible outcomes: 
	\begin{itemize}
		\item $M$ can be reduced infinitely, and in this case $BT(M)$ is simply 
		$\bot$
		\item $M$ can be reduced to a term $N$ in weak head normal form
		$N = f~M_1~M_2~\cdots~M_{ar(f)}$ such that for all $1 \le i \le n$, we 
		have $\p M_i : o$. Then, $BT(M) = 
		f(BT(M_1),BT(M_2),\cdots,BT(M_{\ar(f)}))$.
	\end{itemize}
\end{definition} 
The aforementioned process is infinitary in two different ways: the evaluation 
of $M$ can diverge, and $BT(M)$ can be infinite. The process above is 
well-defined only for closed terms of ground type and is thus less general 
than the one generating B\"ohm Trees for arbitrary terms of the $\la Y$-calculus \cite{Clairambault2013:HORS} 
(there, in particular, one has to deal with $\lambda$-abstractions).

\begin{example}
\label{ex:tree}
As an example, with the signature $\Sigma = \set{g : 2, f : 1, a : 0}$, consider the term
$$M \triangleq (Y~(\la F. \la x. g~x~(F~(f~x))))$$
We pose $M_{\mathit{step}} \triangleq (\la F. \la x. g~x~(F~(f~x)))$. By applying weak head reduction rules, we obtain 
$$M~a \rightarrow M_{\mathit{step}}~M~a \rightarrow (\la x. g~x~(M~(f~x)))~a \rightarrow g~a~(M~(f~a))$$
Thus, $BM(M~a) = g(a,BM(M~(f~a)))$, and similarly $BM(M~(f~a)) = 
g(f(a),BM(M~(f~(f~a))))$, finally we obtain that $BM(M~a) = 
g(a,g(f(a),g(f(f(a)),\cdots)))$ the infinite tree described above.
\end{example}

\subsection{Expressing Properties As Alternating Parity Tree Automata}

Now that we have defined the class of models for higher-order model checking, we can define the specification language. The gold standard in model-checking consists in properties expressed by formulas of monadic second order logic \cite{Knapik2001:MSOforHORS} (\MSO\ for short). However, several equiexpressive specification languages have been defined in the literature, notably $\mu$-calculus \cite{Walukiewicz1993:MuCalculus} and alternating parity tree automata (APT for short) \cite{Emerson91:TreeAutomata,Gradel2003:Automata}. The latter is commonly used in the HOMC literature \cite{Ong2006:HOMC,Ong2015:HOMC,KobayashiOng2009:TypeSystemHOMC}, and as examples we will define through the paper have a simple representation as automata, contrary to \MSO\ formulas, we also chose this specification language.

Given a set $X$ of variables, we define the \emph{positive Boolean formulas} over $X$, denoted $B^+(X)$, with the grammar:
$$ \phi,\psi ::= \mathtt{tt} \midd \mathtt{ff} \midd x \midd \phi \land \psi \midd \phi \lor \psi $$
with $x \in X$. A subset $Y \subseteq X$ satisfies a formula $\phi$, denoted $Y \vDash \phi$ if and only 
if the formula $\phi$ in which all elements $x \in Y$ are replaced by $\mathtt{tt}$ and all elements 
$x \notin Y$ are replaced by $\mathtt{ff}$ is semantically 
true. We can now give the formal definition of an APT automaton:

\begin{definition}[Alternating Parity Tree Automaton]
	An \emph{alternating parity tree automaton} is a tuple $\mathcal{A} = (\Sigma,Q,\delta,q_i,\Omega)$
	where 
	\begin{itemize}
		\item $\Sigma$ is a signature, defining tree constructors with their arity. 
		\item $Q$ is a \emph{finite} set of states, with $q_i \in Q$ the initial state. 
		\item $\delta$ is the transition function, such that for each $f \in \Sigma$, we have $\delta(q,f) \in B^+(\set{0,\dots,\ar(f)-1} \times Q)$. A translation $\delta(q,f)$ is thus a positive formula on a state assignment to children nodes.  
		\item $\Omega : Q \rightarrow \set{0,\dots,M}$ is the priority function, with $M \ge 0$ an integer.  
	\end{itemize}
\end{definition}

To define the acceptance condition of an automaton, we need to introduce the notion of positions in a tree. 

\begin{definition}
	Let $\Sigma$ be a set of tree constructors with maximal arity $A$. The set of all position of a tree $t$, denoted $dom(t)$, is a set of words on the alphabet $\Gamma = \set{0,\cdots,A}$ defined by: 
	$$dom(\bot) = \epsilon \qquad dom(f(t_0,\dots,t_k)) = \set{\epsilon} \cup 0 \cdot dom(t_0) \cup \cdots \cup k \cdot dom(t_k)$$
	where $\cdot$ is the concatenation operator for words. 
	
	For a tree $t$ and a position $\alpha \in dom(t)$ we define $t(\alpha) \in \Sigma \cup \set{\bot}$ by: 
	$$ \bot(\epsilon) = \bot \qquad f(t_1,\dots,t_k)(\epsilon) = f \qquad f(t_1,\dots,t_k)(n \cdot \alpha) = t_n(\alpha) $$
	essentially giving the node of the tree $t$ at position $\alpha$. 
\end{definition}

A \emph{run-tree} of an automaton $\mathcal{A}$ over a tree in $\trees$ is a tree with constructors in 
$dom(t) \times Q$. This run-tree must satisfy the following two constraints:
\begin{itemize}
	\item The root is $(\epsilon,q_i)$, representing the root of $t$ in the state $q_i$. 
	\item For any node $(\alpha,q)$ of the run-tree, there is a set $S \subseteq \set{0,\dots,\ar(t(\alpha))-1}$ which satisfies the transition $\delta(q,t(\alpha))$. And, for each $(i,q') \in S$, one of the child of $(\alpha,q)$ in the run-tree is $(\alpha \cdot i,q')$. 
\end{itemize}
Finally, we say that an infinite branch $(\epsilon,q_i) \cdots 
(\alpha_k,q_k) \cdots$ of the run tree satisfies the 
\emph{parity condition} if the largest priority that occurs 
infinitely often in $\Omega(q_i) \cdots \Omega(q_k) \cdots$ is 
even. A run-tree is \emph{accepting} if every infinite path in 
it satisfies the parity condition. And of course, a tree $t \in \trees$ is 
accepted by an automaton $\mathcal{A}$ if there exists an 
accepting run-tree for $t$.  

\begin{example}
	We define an automaton $\mathcal{A} = (\set{a :2, b : 1, c : 0},\set{q_0,q_1},\delta,q_0,\Omega)$ such that $\mathcal{A}$ accepts a tree $t$ if and only if for every path in $t$, $c$ occurs eventually after $b$ occurs. We pose: 
	\begin{itemize}
		\item For all $q \in \set{q_0, q_1}:$
		$$\delta(q,a) = (0,q) \land (1,q) \qquad \delta(q,b) = (1,q_1) \qquad \delta(q,c) = \mathtt{tt}$$
		Meaning intuitively that a node $a$ propagates the state to both children, seeing the letter $b$ changes the state to $q_1$ and a leaf $c$ is always accepted. 
		\item $$ \Omega(q_0) = 2 \qquad  \Omega(q_1) = 1$$
	\end{itemize} 
	By definition, if $q_1$ appears infinitely often in a branch of the run tree, and $q_0$ does not, then the largest priority is $1$ in this branch, and the tree does not satisfy the parity condition. Thus, with this automata, to accept a tree $t$, it must have only finite branches (that terminates with the only available leaf $c$) or infinite branches in which $b$ does not occur (in order to stay in the state $q_0$). This corresponds indeed to the property that for every path in $t$, $c$ occurs eventually after $b$ occurs".   
\end{example}

We can now define the model-checking problem and give the theorem expressed in \cite{Ong2006:HOMC}. 

\begin{definition}[\MSO\ Model-Checking for Böhm Trees (or APT Acceptance Problem)]
	Given a closed $\la$Y-term $M$ of with $\p M : o$ and an APT $\mathcal{A}$, is $BT(M)$ accepted by $\mathcal{A}$ ?
\end{definition} 

 The groundbreaking result of higher-order model checking \cite{Ong2006:HOMC} can then be summed up with the following theorem:

\begin{theorem}[Decidability \cite{Ong2006:HOMC,Ong2015:HOMC}]
	\MSO\ model-checking on Böhm Trees of closed ground terms of the $\la$Y-calculus is decidable. 
\end{theorem}

The results obtained in \cite{Ong2006:HOMC} give the precise complexity of this problem, which is exponential depending on the \emph{order} of the $\la$Y-term (or, equivalently, the order of the higher-order recursion scheme \cite{NakamuraetAl2020:AverageCaseHardness}). The aim of our work is on decidability, so we will not take orders into account, but we plan to study complexity issues in the future, given that finer complexity analysis for similar problems is already available \cite{NakamuraetAl2020:AverageCaseHardness,Tsukada2014:NonDeterministic}. 

Thanks to this theorem, to show that \MSO\ model-checking is decidable for a class of infinite trees, it is sufficient to show that this class of trees can be computed as Böhm trees generated from $\la$Y-terms, and this can be done via program transformation to the $\la$Y-calculus, as we will see in many occasions starting from the next section. 

\section{Higher-Order Programs with Effects}
\label{s:EPCF}
In this section, we introduce a calculus with algebraic operations and 
fixpoints, called \EPCF. This language can be seen as a fine-grained 
call-by-value variation on Plotkin's \PCF\ endowed with effect-triggering 
operations, as described in 
\cite{Simpson2019:BehaviouralEquivalenceEffects}. We consider a finitary 
version of this language which differs from Simpson and Voorneveld's one only 
in minor ways.

\EPCF\ is built around two syntactic categories, namely \emph{values}, which 
denote data or functions, and \emph{computations}, which are instead programs 
which potentially produce effects when evaluated. Terms and types for \EPCF\  
are as follows:
\begin{align*}
 &\mbox{(Values)} & 	V,W &::= v \midd \underline{n} \midd x \midd \la x. C \midd \fix~x. V \\
&\mbox{(Computations)} & 	C,B,A &::= V~W \midd \return(V) \midd \letin{x = C}~B \midd \sigma(V;x.C) \midd \\
&&&\phantom{::= }\case(V;C_1,\dots,C_k) \\
 &\mbox{(Types)} & 	T,U &::= \type{B} \midd \type{k} \midd T \rightarrow U
\end{align*}
We suppose given a set of \emph{finite} ground types, ranged over by 
$\type{B}$. Examples include the unit type and the type of 
booleans. We similarly assume a set $\mathcal{V}$ of constant values each of a 
ground type, ranged over by $v$, e.g. a unique 
constructor $() \in \mathcal{V}$ with its associated type $\type{Unit}$. 
We also assume finite enumeration types $\type{k}$ to be present, each of them 
inhabited by the values 
$\underline{0},\underline{1},\dots,\underline{k-1}$. We introduce 
enumeration types, which support pattern matching, just to keep them 
distinct from finite base types. 

Algebraic operations have the shape $\sigma(V;x.C)$ where $V$ is said to be a 
\emph{parameter}, and $C$ represents the \emph{continuation} of the 
computation. The 
variable $x$ is bound in $C$ and represents the choice made by the algebraic 
operation. We suppose given a signature set $\Sigma$ of symbols for algebraic 
operations, where each $\sigma \in \Sigma$ is given an arity $\type{B} 
\ararrow k$, meaning that $\sigma$ is an algebraic operation of arity $k$ 
depending on a parameter of type $\type{B}$. In an operation of arity $k$, the type 
of $x$ in the continuation $C$ is $\type{k}$. In particular, defining 
$k$ different continuations $C_1,\dots,C_k$ depending on the possible values of 
$x$ can be done by way of the computation $\sigma(V;x. 
\case(x,C_1,\dots,C_k))$. As common in call-by-value calculi, the fixpoint operator $\fix~x. V$ is a 
value and can only be attributed a function type without limiting the 
expressive power of the language. Finally, the shape of computations is 
restricted in that the only way to compose computations is to use the 
$\mathtt{let}$ constructor. For example, if we want to apply the result of the 
computation $D$ to the function computed by $C$, we need to define a 
computation 
$$\letin{x = D}~\letin{y = C}~y~x \qquad \text{or} \qquad \letin{y = C}~\letin{x = D}~y~x $$
As can be seen from this example, the order of evaluation is imposed by 
sequencing, a property which turns out to be crucial in presence of effects, 
as we will see when looking at the semantics. 

As for the type system of \EPCF, which is pretty natural and standard, 
we give some of the rules in Figure~\ref{f:EPCFtypes}.
\begin{figure}
	\begin{framed}
		\begin{center}
			\AXC{$ v : \type{B} \in \mathcal{V}$}
			\UIC{$\Gamma \p v : \type{B}$}
			\DP 
			\qquad 
			\AXC{$n < k$}
			\UIC{$\Gamma \p \underline{n} : \type{k}$}
			\DP 
			\qquad 
			\AXC{$\Gamma, x : T \rightarrow U \p V : T \rightarrow U$}
			\UIC{$\Gamma \p \fix~x.V : T \rightarrow U$}
			\DP 
			\qquad 
			\AXC{$\Gamma \p V : T$}
			\UIC{$\Gamma \p \return(V) : T$}
			\DP 
			\\
			\vvskip 
			\AXC{$\Gamma \p C : T$}
			\AXC{$\Gamma, x : T \p D : U$}
			\BIC{$\Gamma \p \letin{x = C}~D : U$}
			\DP 
			\qquad 
			\AXC{$(\sigma : \type{B} \ararrow k) \in \Sigma$}
			\AXC{$\Gamma \p V : \type{B}$}
			\AXC{$\Gamma, x : \type{k} \p C : T$}
			\TIC{$\Gamma \p \sigma(V;x.C) : T$}
			\DP 
			\\
			\vvskip 
			\AXC{$\Gamma \p V : \type{k}$}
			\AXC{$(\Gamma \p C_i : T)_{1 \le i \le k}$}
			\BIC{$\Gamma \p \case(V;C_1,\dots,C_k) : T$}
			\DP 
		\end{center}
	\end{framed}
	\caption{Typing Rules for \EPCF}
	\label{f:EPCFtypes}
	\Description{}
\end{figure}
The typing rule for algebraic operations stipulates that an effect does not 
change the type of the whole computation, which is however arbitrary. 
Intuitively this is because calling an algebraic operation induces an effect 
$\sigma$, depending on the parameter $V$, and the computation will continue as 
$C$ after throwing this effect. 

We now introduce the dynamic semantics of \EPCF\ using the reduction relation 
in Figure~\ref{f:EPCFsemantics}.
\begin{figure}
	\begin{framed}
		\begin{center}
			\AXC{}
			\UIC{$(\la x. C)~W \rightarrow C[W/x]$}
			\DP 
			\qquad 
			\AXC{}
			\UIC{$(\fix~x.V)~W \rightarrow (V[(\fix~x.V)/x])~W$}
			\DP 
			\\ \vvskip 
			\AXC{}
			\UIC{$\case(\underline{n},C_0,\cdots,C_{k-1}) \rightarrow C_n$}
			\DP 
			\\ \vvskip 
			\AXC{}
			\UIC{$\letin{x = \return(V)}~C \rightarrow C[V / x]$}
			\DP 
			\qquad 
			\AXC{$C \rightarrow C'$}
			\UIC{$\letin{x = C}~B \rightarrow \letin{x = C'}~B$}
			\DP 
			\\
			\vvskip  
			\AXC{}
			\UIC{$\letin{x = \sigma(V;y.C)}~D  \rightarrow \sigma(V;y. \letin{x = C}~D)$}
			\DP
		\end{center}
	\end{framed}
	\caption{Semantics for \EPCF}
	\label{f:EPCFsemantics}
	\Description{}
\end{figure}
The cornerstone of this relation is the last rule, allowing to algebraic 
operations to ``percolate out'' of non-trivial evaluation contexts. 
Intuitively, this rule means that if we have a computation of the shape 
$\letin{x=C}~D$, and $C$ is a computation that throws an effect $\sigma$, then 
this effect is immediately visible from the outside, its continuation now 
incorporating the $\mathtt{let}$ expression at hand. If we go back to the 
aforementioned example describing the composition of computations, namely the 
pair of terms
$$
\letin{x = D}~\letin{y = C}~y~x \qquad\qquad \letin{y = C}~\letin{x = D}~y~x 
$$
we can see that in the first case, the effects thrown in $D$ appears first, and 
then the one in $C$, and conversely for the other case. 
\begin{example}
	Let us consider a computation which flips a fair coin, then returning the 
	the exclusive or of the obtained value and itself (which is thus always 
	going to be true). All this makes use of an algebraic operation 
	$\mathbf{Flip}$ having arity $\type{unit}\ararrow 2$ and of a binary 
	function $\oplus$ computing the exclusive or, and which can anyway be 
	easily written as a term itself:
	$$
	\letin{x = \mathbf{Flip}(\cdot;y. y)}~(x \oplus x) 
	$$
	In a call-by-value setting, this program should flip a coin, to obtain 
	$\underline{0}$ or $\underline{1}$ thus always returning $\underline{0}$. 
	Indeed, the dynamic semantics allows the $\mathbf{Flip}$ operation to be 
	visible from the outside:
	$$
	\letin{x = \mathbf{Flip}(\cdot;y. y)}~(x \oplus x) \rightarrow 
	\mathbf{Flip}(\cdot;y. \letin{x = y}~ x \oplus x)
	$$
	There are two observations one needs to make now. The first is that the
	continuation $\letin{x = y}~ x \oplus x$ returns $\underline{0}$ for all 
	values of $y$, as expected, but since it occurs as the argument of 
	$\mathbf{Flip}$ it cannot be further reduced, and its evaluation must be 
	somehow be captured differently. The second one is that in 
	presence of effects, confluence is lost. If, indeed, we pass 
	$\mathbf{Flip}(\cdot;y. y)$ unevaluated to the right-hand side, in a 
	call-by-name fashion, we land on a totally different term, namely one which 
	flips \emph{two} coins, possibly returning $\underline{1}$:
	$$
	\letin{x = \mathbf{Flip}(\cdot;y. y)}~(x \oplus x) \rightarrow_{name} 
	(\mathbf{Flip}(\cdot;y. y) \oplus \mathbf{Flip}(\cdot;y. y)).
	$$
\end{example}
The example above implicitly tells us that, similarly to what happens in the 
$\laY$-calculus, the terms of \EPCF\ generate infinite trees that we 
commonly call \emph{effect trees}, but that the dynamic semantics does not by 
itself take care of the unfolding: indeed, typed closed computations in 
normal form in \EPCF\ are either of the shape $\return(V)$, which is expected, 
or $\sigma(V;x.C)$, meaning that evaluation apparently stops as soon as an 
algebraic operation is encountered. The effect tree of a computation is what we 
are looking for as is defined as follows:
\begin{definition}[Effect Trees for \EPCF\ Computations]
	\label{d:effecttrees}
	For an algebraic effect signature $\Sigma$ and a type $T$, let us define
	the signature $\Sigma_{T}$ as follows:
	$$
	\Sigma_{T} = \set{\sigma : k + 1 \mid ((\sigma : \type{B} \ararrow k) \in 
	\Sigma)} \cup \set{\return(V) : 0 \mid\;\p V : T} \cup \set{v : 0 \midd v \in \mathcal{V}}
	$$
	Note that $\Sigma_{T}$ is infinite in the general case, but that it is 
	guaranteed to be finite when $\mathcal{V}$ is finite and $T$ is a ground 
	(and thus finite) type. The effect tree of a typed computation $\p C : T$, 
	denoted $ET(C)$, is a tree in $\etrees{T}$ defined by the following 
	essentially infinite process. First of all, reduce $C$ ad infinitum, and 
	there are \emph{three} possible outcomes: 
	\begin{itemize}
		\item If $C$ can be reduced infinitely without reaching a normal form, then $ET(C) = \bot$
		\item If $C \rightarrow^* \return(V)$ then $ET(C) = \return(V)$
		\item If $C \rightarrow^* \sigma(v;x.C)$ with $\sigma: \type{B} \ararrow k$ then $ET(C) = \sigma(v,ET(C[x:= \underline{0}]),\cdots,ET(C[x := \underline{k-1}]))$ 
	\end{itemize}
\end{definition}

This definition of effect trees differs a bit from the one given in the literature by the use of a special branch for parameters. Formally, standard effect trees \cite{Simpson2019:BehaviouralEquivalenceEffects,MatacheStaton2019:LogicAlgebraicEffects} would be defined with the signature: 
$$\Sigma'_{T} = \set{\sigma_v : k \mid ((\sigma : \type{B} \ararrow k) \in 
\Sigma) \land (\p v : \type{B} \in \mathcal{V})} \cup \set{\return(V) : 0 \mid \p V 
: T}$$ 
with the expected definition when seeing an effect $\sigma(v;x.C))$. Our 
definition is not harmful compared to this usual notion of effect trees in 
which operations are indexed by their parameters because we are in a finite 
case, where the parameters can only be a constant value for some finite type 
$\type{B}$. In particular, any \MSO\ formula on a standard tree with indexed 
effects for can be transferred to an equivalent formula on effect trees with 
branching for parameters, see Section~\ref{s:logic}, 
Proposition~\ref{t:MSObranchingtrees} for more details. Anyway, one can remark that 
this definition is similar to the notion of Böhm trees we defined before. And 
indeed, we show that we can relate the notion of effect trees in \EPCF\ with 
the notion of Böhm trees in $\la Y$-calculus. Formally, we need to use a CPS 
translation similar to the one described in 
\cite{Matache2018:MasterThesis,MatacheStaton2019:LogicAlgebraicEffects}. 

\begin{example}
	As an example, we can compute a tree with similarities with the one of 
	Example~\ref{ex:tree} in the $\la$Y-calculus. We take $\Sigma 
	=\set{\sigma_f : \type{Unit} \ararrow 1, \sigma_g : \type{Unit} \ararrow 
	2}$, and the \EPCF\ computation 
	$$C \triangleq \fix~F.(\la x. \sigma_g((),n. \case(n;,x~(),F (\la z. \sigma_f((),\_.x~()))))$$
	For the sake of clarity, let us ignore all the parameters for algebraic 
	effect, which do not bring anything interesting here because the parameter 
	is always $()$. Then, we have:
	$$C~(\la z. \return(())) \rightarrow \sigma_g((\la z. \return(()))~(),C~(\la z'. \sigma_f((\la z. \return(()))~())~()))$$
	And if we continue the computation, we can see that the effect tree $t$ of $C~(\la z. \return(())) $ is
	$$ t \triangleq \sigma_g(\return(()),\sigma_g(\sigma_f(\return(())),\cdots) $$
	Note that because of the call-by-value paradigm, we need to use thunk 
	functions, something which is not needed in the $\la$Y-calculus. As 
	expected, in the translation from \EPCF\ to $\la$Y-calculus, we have to 
	take into account of all this along the CPS translation. 
\end{example}

\section{From EPCF to the $\lambda Y$-Calculus}
\label{sect:translationEPCFtoY}
It is now time to turn the informal arguments we made at the 
end of the last section formal. Indeed, translating any \EPCF\ 
computation into a term of the $\laY$-calculus is indeed 
possible, as we are going to show. 

The fact that \EPCF\ ground values, arities, and parameters are 
all taken from finite types implies that for every algebraic 
effect signature $\Sigma$, the signature $\Sigma_Y$ defined as 
$\set{v : 0 \mid v \in \mathcal{V}} \cup \set{\sigma : k + 1 
\mid ((\sigma : \type{B} \ararrow k) \in \Sigma)}$ is itself 
finite. This enables the defining of a transformation from 
\EPCF\ to the $\lambda Y$-calculus with constants in 
$\Sigma_Y$, which we first give on types:

\begin{definition}[CPS Type Transform]
	For the sake of conciseness, we use $\neg T$ to denote the type $T \rightarrow o$. We pose 
	$$ \type{B}^* = o \qquad \type{k}^* = \neg o^k  \qquad  (T \rightarrow U)^* = T^* \rightarrow \neg U^* \rightarrow o $$
	We extend this function to contexts, with $\emptyset^* = \emptyset$ and $(\Gamma, x : U)^* = \Gamma^*, x : U^*$.
\end{definition}
Observe how base types and enumeration types are translated in 
two essentially different ways, and that arrows are translated 
following the usual double-negation scheme. The following 
Definition captures how our translation works on \emph{terms}:
\begin{definition}[CPS Term Transformation]
	We define $V^*$ and $C^*$ with the intuition that if $\Gamma \p V : T$ and $\Gamma \p C : U$ then $\Gamma^* \p V^* : T^* $ and $\Gamma \p C^* : \neg \neg U^*$. 
	\begin{align*}
		v^* &= v \text{ when } v \in \mathcal{V} 
		&(\la x. C)^* &= \la x. C^* \\
		x^* &= x 
		&(\fix x.V)^* &= Y (\la x. V^*) \\
		\underline{n}^* &= \la x_0,\dots,x_{k-1}. x_n
		&\case(V;C_1,\dots,C_k)^* &= \la c. V^*~(C_1~c)~\cdots~(C_k~c) \\
		(V~W)^* &= \la c. V^*~W^*~c 
		&(\letin{x = C}~B) &= \la c. C^*~(\la x. B^*~c) \\  
		\return(V)^* &= \la c. c~V^* &&
	\end{align*}
	$$ \sigma(V;x.C)^* = \la c. \sigma~V^*~(C^*[x := \underline{0}^*]~c)~\cdots~(C^*[x := \underline{k-1}^*]~c)$$
\end{definition}

Remark that this translation can be done because of our choice 
of effect trees with a special branch for parameters, indeed it 
is not possible at runtime to know exactly the value of the 
parameter $V$ of an algebraic operation, and with our 
definition of effect trees, we can postpone the identification 
of $V^*$ by just passing it as an argument. This translation is 
well-typed, as shown by the following lemma, whose proof is 
straightforward:
\begin{lemma}
	If $\Gamma \p V : T$ then $\Gamma^* \p V^* : T^* $ and if $\Gamma \p C : U$ then $\Gamma \p C^* : \neg \neg U^*$.
\end{lemma}
Then, we would like to show that this translation is a 
simulation. But, in order to do that, we need to take into 
account that there are some reductions that need to be done 
\emph{outside} weak head reduction in $\lambda$Y, but those 
reductions are actually \emph{administrative} 
\cite{Plotkin94:CPS,Groote1994:CPSTranslation,Hillerstorm2017:Continuation}.
 Formally, we introduce a new reduction relation 
$\rightarrow_{\mathit{adm}}$ by way of the following rules:

\begin{center}
	\AXC{}
	\UIC{$(\la c. M_i)~c \rightarrow_{\mathit{adm}} M_i$}
	\DP 
	\qquad 
	\AXC{$f$ constant}
	\AXC{$(M_i \rightarrow_{\mathit{adm}} N_i)_{1 \le i \le n}$}
	\BIC{$f~M~M_1~\cdots~M_k \rightarrow_{\mathit{adm}} 
	f~M~N_1~\cdots~N_k$}
	\DP 
\end{center}
Intuitively, an administrative reduction here consists in applying a $\beta$-rule, and possibly in the argument of an application (which not possible in head reduction). We take this small set of rules as it is exactly what we need to have the simulation, and it should be intuitive that applying those rules has no consequence in the construction of a Böhm tree of a $\la$Y-term. Remark that we could have defined the transformation of a $\sigma$ using $(\la x. C^*)~\underline{0}^*$ instead of directly $C^*[x := \underline{0}^*]$ if we wanted to avoid using substitution in the encoding. But this would induce more administrative reductions.  

\begin{lemma}
	\label{l:simulationEPCF}
	If $C \rightarrow B$ then for all continuation $c$, $C^*~c 
	\rightarrow^+ M$ such that $B^*~c 
	\rightarrow_{\mathit{adm}}^{\le 2} M$. 
\end{lemma}

\begin{proof}
	The proof can be done by induction on $C \rightarrow B$. 
	All cases are straightforward and do not need 
	administrative reduction except for the rule $\letin{x = 
	\sigma(V;y.C)}~D  \rightarrow \sigma(V;y.\letin{x = C}~D)$, 
	in which we do need administrative reductions: 
	\begin{align*}
	&\la c. \sigma~V^*~((\la c. C^*[y := \underline{0}^*]~(\la 
	x. D^*~c))~c)~\cdots~((\la c. C^*[y := 
	\underline{k-1}^*]~(\la x. D^*~c))~c)~c \\ \rightarrow^{\le 
	2}_{\mathit{adm}} &\sigma~V^*~(C^*[y := 
	\underline{0}^*]~(\la x. D^*~c))~\cdots~(C^*[y := 
	\underline{k-1}^*]~(\la x. D^*~c))
	\end{align*}
\end{proof}
This means in particular that if a computation $C$ never 
reaches a normal form, then $C^*~c$ does not too. Indeed, the 
only case when we use the administrative reduction is the rule 
for a let with an algebraic operation, and if this rule can be 
used in a computation, then the computation will have a normal 
form, beginning with this effect. Then, by comparing the normal 
form of $C$ and $C^*~c$, we obtain the following theorem. 

\begin{theorem}
	For any $\p C : T$, for all $\p c : T^* \rightarrow o$, $BT(C^*~c) = ET(C)[\return(V) \leftarrow BT(c~V^*)]$
\end{theorem} 
Thus, $BT(C^*~c)$ simulates perfectly the effects of $C$, but for the return values, instead of a leaf we have a tree representing the effect applied to this value by the continuation. As a consequence, using standard results of Higher-Order Model Checking, we have:

\begin{corollary}
 \label{t:EPCFdecidability}
 	For any $\p C : T$ and $\p c : T^* \rightarrow o$, \MSO\ model-checking on $ET(C)[\return(V) \leftarrow BT(c~V^*)]$ is decidable.   
\end{corollary}
In particular, if the continuation always return a new constant 
$K$ for example, then $BT(C^*~c)$ is equal to $ET(C)$ with all 
return leaves replaced by $K$. Also, if the type $T$ is a 
ground type, then $c$ can be taken as the identity, and we have 
exactly decidability of \MSO\ model-checking for $ET(C)$. In general, the set of all 
possible values in the form $\return(V)$ for a type $T$ is 
infinite, as they could be $\la$-abstractions for example, but 
the previous theorem gives us a degree of freedom as we can 
always chose the continuation that we want in order to 
distinguish some of those values. 

\section{Specifications Through Observations}
\label{s:logic}
In the previous section, we proved that \EPCF\ computations can 
be translated into equivalent (in the sense of the tree they 
generate) $\laY$-terms. But how can we actually \emph{exploit} 
this result for the sake of verifying interesting properties 
about the effect produced by \EPCF\ programs? Of course, we can 
automatically verify whether any \MSO\ formula has (the effect 
tree generated by) any \EPCF\ term as its model. But, is it 
the end of the story?

As already, mentioned, there are many proposals in the 
literature~\cite{PlotkinPretnar2008:LogicAlgebraicEffects,Simpson2019:BehaviouralEquivalenceEffects,MatacheStaton2019:LogicAlgebraicEffects} about \emph{logics} specifically designed 
for the sake of expressing properties of interest for effectful 
(possibly higher-order) programs. In this 
paper we will somehow follow the path 
recently hinted at by Simpson and 
Voorneveld~\cite{Simpson2019:BehaviouralEquivalenceEffects}. 

Effect trees as we have used them so far are not subject to any 
\emph{equational theory}, the latter often being associated to 
effects 
and justified by a monadic 
interpretation~\cite{Moggi1998:ComputationalMonads}. As an 
example, the following equations hold in the state monad, 
itself a way to interpreting the first example from 
Section~\ref{s:informal}: 
$$ \mathbf{Set}(\ell,\mathtt{true},\_.\mathbf{Get}(\ell,y.C)) = \mathbf{Set}(\ell,\mathtt{true},\_.C[y := \underline{1}]) \quad \mathbf{Set}(\ell,a,\_.\mathbf{Set}(\ell,b,\_.C)) = \mathbf{Set}(\ell,b,\_.C)  $$
As pointed out in 
\cite{Simpson2019:BehaviouralEquivalenceEffects}, 
we can capture equations between effect trees agnostically, 
without the need of explicitly referring to monads, through the 
notion of an \emph{observation}, namely a set of 
\emph{properties} of effect trees:
\begin{definition}[Observations \cite{Simpson2019:BehaviouralEquivalenceEffects}]
	For a given set of algebraic operations $\Sigma$, we define a set of 
	\emph{observations} as a set, denoted $\mathcal{O}$, such that each $o \in 
	\mathcal{O}$ is a subset of the effect trees of type $\type{unit}$. 
\end{definition}
On top of a set of observations $\mathcal{O}$, Simpson and 
Voorneveld define a logic 
endowed with a modal operator of the form $o\phi$, where $o$ is 
an element of $\mathcal{O}$ and 
$\phi$ is a formula. An effect tree satisfies $o\phi$ if it is 
among the trees in $o$ \emph{and} all its leaves satisfy 
$\phi$. This way, the logic can \emph{observe} the tree at 
hand through the lenses of the observations in $\mathcal{O}$, 
implicitly accounting for equations: it is necessary to define 
$\mathcal{O}$ such that all its elements are invariant by the 
desired set of equations. The resulting logic is 
quite powerful, being infinitary in nature, and turns out to 
precisely capture observational equivalence 
\cite{DalLago2017:EffectfulBisim} for (a calculus closely 
related to) \EPCF. Noticeably, most notions of effects, 
including nondeterministic, probabilistic and state effects can 
be captured this 
way~\cite{Simpson2019:BehaviouralEquivalenceEffects} 

Unfortunately, however, the aforementioned logic is strictly 
more powerful than \MSO\ or the $\mu$-calculus.  
This happens for two distinct reasons: on the one hand, it is 
essentially infinitary in nature, this way accounting for the 
fact that the target calculus has \emph{countable} types. On 
the other hand, the underlying set of observations 
$\mathcal{O}$ can be arbitrarily complex, thus injecting a huge 
discriminating power into the logic, something that \MSO\ 
cannot do natively.

Motivated by all that, we address the following question in the 
rest of this section: can we express useful properties of 
effect trees as observations \emph{from within} \MSO\ itself? In 
other words, is it the case that for a given notion of effect, all the 
elements of $\mathcal{O}$ are captured by \MSO\ formulas? This 
way, we could at least be sure, given the results from 
Section~\ref{sect:translationEPCFtoY}, that model-checking is 
decidable \emph{relative to} $\mathcal{O}$. 
This is precisely what motivates the following definition of 
an extension of \MSO\ (or $\mu$-calculus) with observation.

\begin{definition}[Observation Predicate]
	Suppose given a signature $\Sigma$ for effect trees, containing a set of algebraic operations, and a set of other possible 
	constructors (that could come, e.g., from parameter values, 
	or special constants of the $\la$Y-calculus). Fix a set of 
	observations $\mathcal{O}$.
	Let $t \in \trees$, and $P$ a subset of all the possible constructors of arity $0$. We define 
	$t[\in P]$ as the tree obtained from  $t$ by replacing the 
	leaves in $P$ with $\return(())$ and all the other leaves 
	with $\bot$. We next define, for each $o \in \mathcal{O}$ 
	and each $P$, a formula $o_P$, called an observation 
	predicate, such that for any effect tree t, we have 
	$$ t \vDash o_P \text{ iff } t[\in P] \in o$$
	Note that we have indeed that $t[\in P]$ is an effect tree 
	of type $\type{Unit}$, thus it can be an element of $o$.
	Write $\MSOp{\mathcal{O}}$ for the conservative extension 
	of \MSO\ (on infinite trees) obtained by enriching the class of 
	formulas with observation predicates.
\end{definition}
This notion of observation predicate having access to a set $P$ 
is, thanks to finiteness, equiexpressive to the formulas $o 
\phi$ defined in 
\cite{Simpson2019:BehaviouralEquivalenceEffects}, in which 
$\phi$ is supposed to be a formula that separates the correct 
leaves from the incorrect ones in an \emph{infinite} set of 
leaves. Here, because of finiteness, we have a finite 
number of such separations, thus we can take a set $P$ instead 
of a formula $\phi$.  

$\MSOp{\mathcal{O}}$ can thus be seen as a reasonable logic for 
writing down specifications about higher-order effectful 
programs. But is the model checking problem emerging out of it 
decidable? There are at least two possible answers to this 
question: 
\begin{itemize}
\item 
	Either the observation predicates from $\mathcal{O}$ are 
	definable in \MSO\ \emph{itself}, in this case 
	$\MSOp{\mathcal{O}}$ is equiexpressive with \MSO, and the 
	model checking problem remains decidable.
\item 
	Or observation predicates are not \MSO-definable, and we 
	cannot conclude, with the concrete risk of landing into an 
	undecidable verification problem.
\end{itemize}

We will now see some cases of effects and observations, each of 
them corresponding to one of the two cases above. Let us start 
from an example of effect which is hard to capture.

\begin{example}[Probabilistic Observations]
	The typical counter-example of effects to \MSO-definability 
	of observation predicates comes from probabilistic 
	effects. The set of observations for probabilistic choice, 
	with a unique 
	operation $\mathbf{Flip} : \type{unit} \ararrow 2$, is given by 
	$\mathcal{O} = \set{o_{>q} \midd q \in \mathbb{Q}, 0 \le q 
	< 
	1}$.
	Essentially, an observation predicate $o_{>q}$ means that 
	the probability that 
	the program terminates is greater than $q$. This is  
	\emph{not} a property which can be defined in \MSO\ on 
	effect trees, and this happens for very good reasons, since 
	we know that higher-order model checking for probabilistic 
	HORSs is in general undecidable 
	\cite{Kobayashi2020:Termination}). 
\end{example}

Before going into positive examples, we need to solve a small discrepancy between the literature on algebraic effects, especially with the standard definition of effect trees \cite{Simpson2019:BehaviouralEquivalenceEffects}. Indeed, as explained in Section~\ref{s:EPCF}, the Definition~\ref{d:effecttrees} that introduces effect trees in our work is slightly different from the standard notion of effect tree, which was essential for the translation of Section~\ref{sect:translationEPCFtoY}. However, in practice we would like to work directly with standard effect trees, as we did in Section~\ref{s:informal}, since they are simpler and corresponds to the literature. Concretely, we can show that our definition is not harmful for the question of \MSO\ model-checking, since we can transfer any \MSO\ properties (or APT) on standard effect trees to properties on effect trees with a special branch for parameters. Formally, we have:

\begin{definition}[Adding a branch for parameters]
	For any tree $t$ defined with the signature 
	$$\Sigma' = \set{\sigma_v : k \mid ((\sigma : \type{B} \ararrow k) \in \Sigma) \land (\p v : \type{B} \in \mathcal{V})} \cup S$$
	where $S$ is a set of possible other constructors (coming from a chosen continuation in the $\la$Y-calculus according to Corollary~\ref{t:EPCFdecidability} for example), we define the tree $t_p$ with the signature
	$$\Sigma = \set{\sigma : k + 1 \mid ((\sigma : \type{B} \ararrow k) \in \Sigma)} \cup \set{v : 0 \midd v \in \mathcal{V}} \cup S$$
	as the tree $t$ in which all nodes $\sigma_v(t^1,\cdots,t^k)$ are replaced by the tree $\sigma(v,t^1_p,\cdots,t^k_p)$. This translation corresponds exactly to the translation of a standard effect tree to the one we introduced in Definition~\ref{d:effecttrees}.
\end{definition}

We can show that any formula (or APT) that is definable in standard effect tree in $\treess$ is definable for the effect trees in $\trees$ of Definition~\ref{d:effecttrees}. 

\begin{proposition}[Effect Trees and APT]
	\label{t:MSObranchingtrees}
	 
	For any APT $\mathcal{A}$ that recognizes \emph{standard} effect trees in $\treess$, there exists an APT $\mathcal{A}_p$ that recognizes the set $\set{t_p \in \trees \midd t \text{ is accepted by } \mathcal{A}}$
\end{proposition}

\begin{proof}
	The proof consists in adding a special state $q_v$ for each variable $v \in \mathcal{V}$ and define transition as a big \emph{or} depending on the value of the parameter. 
	Formally, we define $\mathcal{A}_p = (\Sigma,Q_p,\delta_p,q_i,\Omega_p)$ such that:
	\begin{itemize}
		\item $Q_p \triangleq Q \sqcup \set{q_v \midd v \in \mathcal{V}}$
		\item $\Omega_p(q) = \Omega(q)$ for $q \in Q$ and $\Omega_p(q_v) = 0$.
		\item For the transition function we have:
		\begin{itemize}
			\item $\delta_p(q,s) = \delta(q,s)$ if $s \in S$  
			\item For all $\sigma : \type{B} \ararrow k$, we have $\delta_p(q,\sigma) = \bigvee\limits_{v : \type{B} \in \mathcal{V}} (\delta(q,\sigma_v)[+1] \land (0,q_v))$ where $\phi[+1]$ is the formula $\phi$ in which all couples $(n,r)$ are
			replaced by $(n+1,r)$.
			\item For all $v \in \mathcal{V}$, $\delta_p(q_v,v) = \mathtt{tt}$ and $\delta_p(q_v,f) = \mathtt{ff}$ for all $f \ne v$.  
		\end{itemize}
	\end{itemize}
		Thus, intuitively, when we reach a node $\sigma$, only one of all the $(0,q_v)$ will be satisfied, the one corresponding to the actual parameter $v$ of $\sigma$. And then, the only satisfiable formula in this big or would be $\delta(q,\sigma_v)[+1] \land (0,q_v)$, meaning that the transition is logically equivalent to $\delta(q,\sigma_v)[+1]$. Thus, we can see that the new transition $\delta_p(q,\sigma)$ is equivalent to the actual formula $\delta(q,\sigma_v)[+1]$ where $v$ is the parameter of $\sigma$ in the tree $t_p$, and from this it is straightforward to see that $\mathcal{A}_p$ recognizes the set $\set{t_p \midd t \text{ is accepted by } \mathcal{A}}$. 
\end{proof}

Let us now turn to examples of observations which are indeed \MSO-definable. 

\begin{example}[Non-determinism]
	Let us consider the unique effect $\set{or : \type{Unit} \ararrow 2}$. The notion of observations for non-determinism is given by the two modalities $\diamond$ and $\square$, where $\diamond_P$ is the set of effect trees that has at least one finite branch with a leaf in $P$, and $\square_P$ is the set of effect trees that has only finite branches, and all of them has a leaf in $P$. 
	
	Those observations are definable by an APT. For the sake of the example, we explain informally the encoding of $\square_P$. We take a unique state $q_i$. To impose finite branching everywhere, it is sufficient to pose $\Omega(q_i) = 1$, because with this no infinite branch satisfies the parity condition. Then, We can just explore through the tree, with the transition $\delta(q_i,or) = (0,q_i) \land (1,q_i)$. And finally, we can separate leaves with $\delta(q_i,a) = \mathtt{tt}$ if $a \in P$ and $\delta(q_i,a) = \mathtt{ff}$ if $a \notin P$ for all leaves $a$. 
\end{example} 

\begin{example}[Finite State Monad]
Going back to the examples used in Section~\ref{s:informal}, we start with the state monad, where the
effects are $\mathbf{Get}$ and $\mathbf{Set}$, with finite domains. As expressed in \cite{Simpson2019:BehaviouralEquivalenceEffects},
the set of observations for this effect is defined by 
$$ \mathcal{O} = \set{o_{p \mapsto q} \midd p,q \text{ are memory states}} $$ 
with the intuitive semantics that $o_{p \mapsto q}$ contains all effect trees 
mapping the state $p$ to the state $q$ (this computation is deterministic 
because the initial state is fixed, so we always know which branch should be 
taken on a $\mathbf{Get}$).
It is easy to see that the observations predicates coming from this set are 
definable in an APT, because an observation predicate $(o_{p \mapsto q})_P$ can 
be represented by an automaton starting from the state $p$, that accepts only when 
the unique path leads to a leaf $x \in P$ in the state $q$, similarly to the automata
defined in Section~\ref{s:informal}. Thus, as expected, 
we can deduce that many interesting properties can be decided on programs using 
a finite state monad, since the observations are \MSO-definable. 
\end{example}

\section{A Calculus of Effects and Handlers}
\label{s:handlers}
An algebraic operation in a language such as \EPCF\ comes out 
\emph{uninterpreted}: a 
computation produces an effect tree, in which all algebraic operations are tree 
constructors, and no meaning is given to operations in this tree. The notion of 
observation can indeed be seen as a way to at least identify trees which should 
be considered the same, implicitly capturing equations between them. 
Observations, however, attribute a \emph{static} meaning to programs, and the 
programmer has no control over them. There is however yet another way of 
attributing meaning to effects in which the 
programmer has direct control on what happens. This can be seen as an 
abstraction on the notion of an exception handler, but 
also of, e.g., the \HASKELL's monad construction. We are referring of course to 
the so-called \emph{effect handlers} \cite{KammarICFP2013:Handlers}. 

In this section, we present a calculus obtained by endowing \EPCF\ with 
effect handlers, called \HEPCF, whose definition closely follows the 
literature on the subject~\cite{KammarICFP2013:Handlers,Hillerstorm2017:Continuation}. This is necessary to even understand what 
doing model checking actually means in presence of handlers, this way paving 
the way for positive and negative results about it, which are deferred to the next 
Section.

The language \HEPCF\ differs from \EPCF\ both at the level of terms and at the 
level of types. The former are extended with the $\mathtt{handle}$ clause, in 
which a computation is evaluated in a protected environment, in such a way as 
to be 
able to capture (and handle) the effects it raises. The latter are enriched 
with annotations keeping track of which ones among the operations are visible. 
The syntax of \HEPCF\ is defined as follows:
\begin{align*}
 &\mbox{(Types)}&	T,U &::= \type{B} \mid k \mid T \rightarrow_E U \\
 &\mbox{(Effects)}&	E &::= \emptyset \midd \set{\sigma : \type{B} \ararrow k} \cup E \\
 &\mbox{(Values)}&	V,W &::= v \midd \underline{n} \midd x \midd \la x. C \midd \fix~f. V \\
 &\mbox{(Computations)}&	C,D &::= V~W \midd \sigma(V; x.C) \midd \return(V) \midd \letin{x = C}~D \\
 &&&\phantom{::=} \midd  \handle{H}~C \midd \case(V;C_1,\dots,C_k) \\ 
 &\mbox{(Handlers)}&	H &::= \set{\return(x) \mapsto C_r} \cup \set{\sigma_i(x;r) \mapsto C_i \mid 1 \le i \le n}
\end{align*}

We add a new information in the type $T \rightarrow_E U$, denoting the set $E$ 
of available algebraic operation the function at hand can possibly perform once 
applied to an argument of type $T$. An handler consists of a return clause 
$\return(x) \mapsto C_r$ and one clause $\sigma_i(x;r) \mapsto C_i$ for each 
handled operation $\sigma_i$ in some effect set $E$. The intuition behind such 
an handler is that it takes a computation using the effects in $\set{\sigma_i 
\mid 1 \le i \le n}$, and interprets each call to the algebraic operation 
$\sigma_i$ as the computation $C_i$ in which the continuation is passed as a 
variable $r$. The return computation $C_r$ is called when the initial 
computation returns a value.
 
Selected rules of the type system are in Figure~\ref{f:HEPCFtyping}. A judgment 
attributing a type to a computation $C$ now has the shape $\Gamma \p_E C : T$, 
meaning that $C$ has the type $T$ in the effect context $E$. 
\begin{figure}
	\begin{framed}
		\begin{center}
			\AXC{$\Gamma, x : T \p_E C : U $}
			\UIC{$\Gamma \p \la x.C : T \rightarrow_E U$}
			\DP 
			\qquad 
			\AXC{$\Gamma, f : T \rightarrow_E U \p V : T \rightarrow_E U $}
			\UIC{$\Gamma \p \fix~f.V : T \rightarrow_E U$}
			\DP 
			\\
			\vvskip  
			\AXC{$\Gamma \p V : T \rightarrow_E U$}
			\AXC{$\Gamma \p W : T$} 
			\BIC{$\Gamma \p_E V~W : U$}
			\DP 
			\qquad   
			\AXC{$(\sigma : \type{B} \ararrow k) \in E$}
			\AXC{$\Gamma \p V : \type{B}$}
			\AXC{$\Gamma, x : k \p_E C : T$}
			\TIC{$\Gamma \p_E \sigma(V;x.C) : T$}
			\DP 
			\\
			\vvskip  
			\AXC{$\Gamma \p V : T$}
			\UIC{$\Gamma \p_E \return(V) : T $}
			\DP 
			\qquad
			\AXC{$\Gamma \p_E C : U$}
			\AXC{$\Gamma, x : U \p_E D : T $} 
			\BIC{$\Gamma \p_E \letin{x = C}~D : T $}
			\DP
			\\ 
			\vvskip 
			\AXC{$\Gamma \p V : k$}
			\AXC{$(\Gamma \p_E C_i : T)_{1 \le i \le k}$}
			\BIC{$\Gamma \p_E \case(V;C_1,\dots,C_k) : T$}
			\DP
			\qquad  
			\AXC{$\Gamma \p H : U _{E'}\Rightarrow_E T$}
			\AXC{$\Gamma \p_{E'} C : U$}
			\BIC{$\Gamma \p_E \handle{H}~C : T$}
			\DP 
			\\ 
			\vvskip  
			\AXC{$E' = \set{\sigma_i : \type{B}_i \ararrow k_i \mid 1 \le i \le n}$}
			\AXC{$\Gamma, x : U \p_E C_r : T $}
			\AXC{$\Gamma, x : \type{B}_i, r : k \rightarrow_E T \p_E C_i : T$}
			\TIC{$\Gamma \p \set{\return \mapsto C_r} \cup \set{\sigma_i(x;r) \mapsto C_i \mid 1 \le i \le n} : U _{E'}\Rightarrow_E T$}
			\DP 
		\end{center}
	\end{framed}
	\caption{Static Semantics of HEPCF}
	\label{f:HEPCFtyping}
	\Description{}
\end{figure}
As the reader can see in the rule for a call to an algebraic operation, a 
computation of type $\Gamma \p_E C : T$ can only give rise to the operations in 
$E$. As expected, a handler takes a computation $C$ of some type $U$ that uses 
the effects in $E'$ and transforms this computation into a computation of type 
$T$ using the effects in $E$. This operation is very similar in principles to 
an application, that is why we describe the type of an handler with an arrow 
type. In details, the type $U _{E'}\Rightarrow_E T$ means that all the 
computations defined in the typed handler should use only the effect in $E$, 
and the handled operations are exactly those in $E'$. In such a context, a 
continuation for an algebraic operation is given as the variable $r$ of type $k 
\rightarrow_E T$, meaning that this continuation is, as expected, of type $T$ 
with effects in $E$, and there are $k$ different continuations of this type, 
corresponding to the $k$ possible branches of this algebraic operation. 

Finally, the dynamic semantics of \HEPCF\ is described in 
Figure~\ref{f:HEPCFSemantics} and Figure~\ref{f:EPCFsemantics}.
\begin{figure}
	\begin{framed}
		\begin{center}
			\AXC{$C \rightarrow C'$} 
			\UIC{$\handle{H}~C \rightarrow \handle{H}~C' $}
			\DP 
			\\ 
			\vvskip  
			\AXC{$H \equiv \set{\return(x) \mapsto C_r} \cup \set{\sigma_i(x;r) \mapsto C_i \mid 1 \le i \le n}$}
			\UIC{$\handle{H}~\return(V) \rightarrow C_r[x := V]$} 
			\DP 
			\\
			\vvskip 
			\AXC{$H \equiv \set{\return(x) \mapsto C_r} \cup \set{\sigma_i(x;r) \mapsto C_i \mid 1 \le i \le n}$}
			\UIC{$\handle{H}~\sigma_i(V;x.C) \rightarrow C_i[x := V, r := (\la x. \handle{H}~C)] $}
			\DP 
		\end{center}
	\end{framed}
	\caption{Semantics for HEPCF}
	\label{f:HEPCFSemantics}
	\Description{}
\end{figure}
As expected, handling a returned value consists in calling the return clause, 
and for the case of an algebraic operation, it is important to note that the 
handler acts also on the continuation. As for the definition of an effect tree, 
we can see that as before, a closed typed computation $\p_E C : T$ in normal 
form is either an algebraic operation in $E$ or $\return(V)$, thus we can 
define effect trees similarly to Definition~\ref{d:effecttrees}, with the 
signature $E$ instead of $\Sigma$. Note that we can see the typing $\p H : T 
_{E'}\Rightarrow_E U$ as defining a tree operation from effect trees on $E'$ of 
type $T$ to effect trees on $E$ of type $U$. In practice, this transformation 
takes a tree defined by $\p_{E'} C : T$ and gives a tree defined by $\p_E 
\handle{H}~C : U$. We will see in the next section that, unfortunately, the 
tree transformations that can be expressed by handlers form a very large class, 
definitely too broad for our purposes.

\section{Model Checking Handled Effects: Positive and Negative Results}
\label{s:modelcheckinghandlers}
In the previous section, we introduced effect handlers, a very powerful and 
elegant linguistic construction by which the interpretation of algebraic 
operations can somehow be delegated to the programmer. In this section that 
this notion of handlers is too expressive for our purposes, meaning that \MSO\ 
model-checking of effect trees produced by terms of \HEPCF\ is \emph{not} 
decidable. We will also show that when handlers have a restricted but 
non-trivial form model checking becomes decidable.

\subsection{Undecidability Through the Halting's Problem}
The proof of undecidability is structured around the encoding of Plotkin's 
\PCF\ into \HEPCF. The encoding is not trivial, because \PCF\ is taken in its 
full generality (thus including a type of natural numbers) while \HEPCF\ only 
has \emph{finite} base types. Since the halting problem for \PCF\ is well-known 
to be undecidable, and can be written down as an \MSO\ formula, undecidability 
of the HOMC problem for \HEPCF\ easily follows. Formally, we introduce standard 
\PCF\ as the following language:
\begin{definition}[Plotkin's \PCF]
	The \PCF\ language (in fine-grained call-by-value) is defined by the following grammar:
	\begin{align*}
	&\mbox{(Types)}&	T &::= \type{Nat} \mid T \rightarrow U \\
	&\mbox{(Values)}&	V,W &::= 0 \midd \mathtt{succ}(V) \midd x \midd \la x. C \midd \fix~f. V \\
	&\mbox{(Computations)}& C,D &::= V~W \midd \return(V) \midd \letin{x = C}~D \midd \case(V;0 \mapsto C;\mathtt{succ}(x) \mapsto D) 
	\end{align*}
\end{definition}
This is a well-known language for which the halting problem is not decidable: 
encoding all partial recursive functions is an easy exercise. We omit the 
definition of the static and dynamic semantics, which are anyway natural and 
very standard.

In order to prove the undecidability of \HEPCF\, we show a translation from 
\PCF\ to \HEPCF. Again, note that the crucial difference between the two 
languages is that \PCF\ has access to an infinite type $\type{Nat}$ for natural 
numbers, with a pattern matching constructor, whereas in \HEPCF\ there are 
handlers and algebraic effects but all base types must be finite. 
But how could we encode the infinite type of natural numbers using finite 
types, effects and handlers? In order to do this, we introduce the effect $E = 
\set{\sigma : \type{Unit} \ararrow 1}$. For the sake of simplicity, we will 
ignore the parameter for this operation, of type $\type{Unit}$. Intuitively, we 
will represent a natural number $n$ by a computation in which this operation 
$\sigma$ is called exactly $n$ times, and then the computation halts by 
returning the only inhabitant of the type $\type{Unit}$. But, because a natural 
number $n$ is supposed to be a value, we will use a thunk function $f_n$ along 
the encoding, and each time we need to inspect the value of a natural number 
$n$, we can call $f_n()$ to produce the effect tree corresponding to the tree 
representation of this number $n$. 

\begin{definition}[Translation from PCF to \HEPCF]
For any type $T$ of \HEPCF\, we define the \emph{zero} of this type, denoted 
$\mathtt{Z}_T$ as follows:
$$ \mathtt{Z}_{\type{Unit}} = () \qquad \mathtt{Z}_k = \underline{0} \qquad \mathtt{Z}_{T \rightarrow_E U} = \la x. \return(\mathtt{Z}_U)   $$
Intuitively, the zero of a type is just a particular closed value of this type 
chosen arbitrarily, in which no effects are used. For the sake of clarity, we 
introduce a notation for several $\la$-abstractions and several applications in 
a row:
$$ \la f,g. C \triangleq \la f. \return(\la g. C) \qquad V~W_1~W_2 \triangleq \letin{x = V~W_1}~(x~W_2)$$
for which we obtain, as expected, that 
$(\la f,g. C)W_1~W_2 \rightarrow^+ C[f := W_1][g := W_2]$. The translation 
$\sem{\cdot}$ from \PCF\ to finitary \HEPCF\ is then given by
$$ \sem{\type{Nat}} = \type{Unit} \rightarrow_E \type{Unit} \qquad \sem{T \rightarrow U} = \sem{T} \rightarrow_E \sem{U}$$
\begin{align*}
	\sem{0} &\triangleq \mathtt{Z}_{\sem{Nat}} 
	&\sem{\mathtt{succ}(V)} &\triangleq \la x. \sigma(y. \sem{V} ()) \\
	\sem{x} &\triangleq x 
	&\sem{\la x. C} &\triangleq \la x. \sem{C} \\ 
	\sem{\fix f. V} &\triangleq \fix f. \sem{V}  
	&\sem{V~W} &\triangleq \sem{V}~\sem{W} \\ 
	\sem{\return(V)} &\triangleq \return(\sem{V}) 
	&\sem{\letin{x = C}~D} &\triangleq (\letin{x = \sem{C}}~\sem{D}) 
\end{align*}   
$$\sem{\case(V;0 \mapsto C;\mathtt{succ}(n) \mapsto D)} \triangleq \letin{a = (\handle{H(C,D)}~\sem{V}~())}~(a~(\la f,g. g~()))$$
where $H(C,D)$ is defined, when $C$ and $D$ have type $T$, by: 
\begin{align*}
	\return(V) &\mapsto \return(\la p. p~\mathtt{Z}_{\type{Unit} \rightarrow \type{Unit}}~(\la x. \sem{C})) \\
	\sigma(r) &\mapsto \letin{n = \return(\la x. \letin{y = r~\underline{0}~(\la f,g. \letin{z = f~()}~\mathtt{Z}_{\sem{T}})}~\return(\mathtt{Z}_{\type{Unit}}))} \\
	&\phantom{\mapsto }~ \return(\la p. p~(\la x. \sigma(z. n~()))~(\la x. \sem{D})) 
\end{align*}
\end{definition}

Let us now give some hints about the translation. As expected, zero is just the 
function that return $()$ when called, and a successor is a call to the effect 
$\sigma$. In order to do the pattern matching for a natural number $V$ given as 
a thunk function $\sem{V}$, we produce the effect tree representing $n$, by 
calling $\sem{V}()$, and we handle it. This handler produces a pair of thunk 
computations. The first computation is just a copy of the initial computation 
$\sem{V}()$. This is rather obvious for the return clause, and for the case of 
$\sigma$, the value $n$ corresponds intuitively to just taking the first 
computation in $r$, thus a copy of the predecessor, and in the returned pair of 
computation, the first pair is $\sigma$ applied to this $n$, which gives the 
successor of the predecessor, i.e. a copy of the current number. The second 
computation simulates the initial pattern matching. For the case of a zero, the 
second computation is just $C$, and for the case of a successor (an operation 
$\sigma$), the second computation takes the copy of the predecessor given by 
$n$, and then returns the computation $D$ that uses this predecessor. 

The proof that this translation is correct is inspired by the proof of the 
encoding of shallow handlers using deep handlers, described in 
\cite{Hillerstrom2018:Shallow}. Intuitively, this is because a shallow handler 
is a handler that only handles the root of an effect tree, which is exactly 
what a pattern matching here is supposed to do. First, we can show that this 
translation is well-typed:

\begin{lemma}
	If $\Gamma \p V : T $ then $\sem{\Gamma} \p \sem{V} : \sem{T}$ and if $\Gamma \p C : T $ then $\sem{\Gamma} \p_E \sem{C} : \sem{T}$
\end{lemma}

\begin{proof}
	The proof is direct except for the $\case$ construct, where we need to 
	prove this intermediate result:
	
	If $\Gamma \p C : T$, $\Gamma, n : \type{Nat} \p D : T$, $\sem{\Gamma} \p \sem{C} : \sem{T}$ and $\sem{\Gamma}, n : \sem{\type{Nat}} \p \sem{D} : \sem{T}$, then 
	$$\sem{\Gamma} \p H : \type{Unit} _E \Rightarrow_E ((\type{Unit} \rightarrow \type{Unit}) \rightarrow (\type{Unit} \rightarrow \sem{T}) \rightarrow \sem{T}) \rightarrow \sem{T}$$
	With this type in mind, the proof that $H(C,D)$ respects this type is straightforward, and we can conclude that the translation for $\case$ is well-typed.
\end{proof}

We want to prove that the translation is a correct simulation. In order to do 
this, we need to introduce formally an approximation of a term up to some 
administrative context, as in \cite{Hillerstrom2018:Shallow}, because in this 
translation we encode integers as functions, and we will sometimes have a term 
that is overly complex because all the computation is hidden behind a $\la$, 
but this term will behave exactly like another simple term corresponding to the 
encoding of an integer. This is exactly what happens here; the variable $n$ in 
the handler is essentially a copy of the predecessor, but syntactically $n$ is 
described as a complex computation.  

\begin{definition}[Administrative Contexts (See \cite{Hillerstrom2018:Shallow},Definition~6)]
	Evaluation contexts are defined with the following grammar:
	\begin{align*}
		\mathcal{E} ::= [] \midd \letin{x = \mathcal{E}}~C \midd \handle{H}~\mathcal{E}
	\end{align*}
	They corresponds exactly to context in which a computation can be reduced. An evaluation context $\mathcal{E}$ for a computation is called \emph{administrative}, denoted adm($\mathcal{E}$) if and only if:
	\begin{itemize}
		\item For all values $V$, $\mathcal{E}[\return(V)] \rightarrow^* \return(V)$
		\item For all operations $\sigma$, $\mathcal{E}[\sigma(V;x.C)] \rightarrow^* \sigma(V;x.D)$ with $D \rightarrow^* \mathcal{E}[C]$.  
	\end{itemize}
\end{definition} 

Remark that the composition of two administrative contexts is also an administrative context. 

\begin{definition}[Approximation up to Administrative Reductions (See \cite{Hillerstrom2018:Shallow},Definition~7)]
	We define $\ge$ for terms as the closure under \emph{any} context of the following rule:
	\begin{center}
		\AXC{}
		\UIC{$C \ge C$}
		\DP 
		\qquad 
		\AXC{$adm(\mathcal{E})$}
		\AXC{$C \ge D$}
		\BIC{$\mathcal{E}[C] \ge D$}
		\DP 
		\qquad
		\AXC{$C \rightarrow C'$}
		\AXC{$C' \ge D $}
		\BIC{$C \ge D$}
		\DP 
	\end{center}
	Intuitively, we have $C \ge D$ when $C$ could be reduced to $D$ using strong rules of reduction (reducing under any context) without changing normal forms (of the form $\return$ or $\sigma$).   
\end{definition}

We can then prove the following lemma on this approximation, showing that the properties of administrative context are valid for approximations. 

\begin{lemma}
	If $C' \ge C$ then:
	\begin{itemize}
		\item If $C$ is $\return(V)$, then $C' \rightarrow^* \return(V')$ with $V' \ge V$
		\item If $C$ is $\sigma(V;y.B)$, then $C' \rightarrow^* \sigma(V';y. B')$ with $B' \ge B$
	\end{itemize}
\end{lemma} 

\ifLong 
\begin{proof}
	We proceed by induction on $C' \ge C$. All the contextual rules are straightforward because $C'$ has directly the desired shape in this case. 
	\begin{itemize}
		\item If $C' \equiv \mathcal{E}[C'']$ with $C'' \ge C$, then we have, by induction hypothesis, if $C \equiv \return(V)$ then $C' \rightarrow^* \mathcal{E}[\return(V')] \rightarrow^* \return(V')$ with $V' \ge V$ since $\mathcal{E}$ is administrative. Similarly, if $C \equiv \sigma(V;y.B)$ then $C' \rightarrow^* \mathcal{E}[\sigma(V';y. B')] \rightarrow^* \sigma(V', y. B'')$ with $B'' \rightarrow^* \mathcal{E}[B']$ and $B' \ge B$, thus $B'' \ge B$. Moreover, $V' \ge V$. 
		\item If $C' \rightarrow C''$ with $C'' \ge C$, then this is direct by induction hypothesis. 	
	\end{itemize}
\end{proof}
\else 
\begin{proof}
	The proof is done by induction on $C' \ge C$ and can be found in the supplementary material.   
\end{proof}
\fi 

We can then show that this approximation is a simulation. 

\ifLong 
\begin{lemma}
	If $C' \ge C$ and $C \rightarrow D$ then there exists $D'$ such that $D' \ge D$ and $C' \rightarrow^+ D'$
\end{lemma}

\begin{proof}
	We proceed by induction on $C' \ge C$.
	\begin{itemize}
		\item If $C'$ is $\mathcal{E}(C'')$ with $C'' \ge C$ and $adm(\mathcal{E})$. Then, by induction hypothesis, there exists $D''$ with $D'' \ge D$ and $C'' \rightarrow^+ D''$. We pose $D' \triangleq \mathcal{E}(D'')$. Because $\mathcal{E}$ is an evaluation context, we have $\mathcal{E}(C'') \rightarrow^+ \mathcal{E}(D'')$, and because $\mathcal{E}$ is administrative, we have $\mathcal{E}(D'') \ge D$. This concludes this case. 
		\item If $C' \rightarrow C''$ with $C'' \ge C$. Then by induction hypothesis, there exists $D''$ such that $C'' \rightarrow^+ D''$ with $D'' \ge D$, and we can pose $D' \triangleq D''$.  
		\item If $C'$ is $V'~W'$ with $V' \ge V$ and $W' \ge W$, and $C = V~W$. Then, there are two cases to consider depending on the reduction  $C \rightarrow D$.
		\begin{itemize}
			\item $V = \la x. E$ for some computation $E$. Then, by definition, $V' = \la x. E'$ with $E' \ge E$. The reduction gives us $C \rightarrow E[x := W]$. But, we have $C' \rightarrow E'[x := W']$, and $E'[x := W'] \ge E[x := W ]$. This concludes this case. 
			\item The case $V = \fix f. U$ is similar to the previous one, we take the obvious reduction and it concludes. 
		\end{itemize}
		\item $C'$ cannot be of the shape $\return(V)$ or $\sigma(V;x.B)$ because it would mean that $C$ is in normal form, and thus cannot be reduced. 
		\item If $C \equiv \letin{x = C_1}~C_2$ and $C' \equiv \letin{x = C_1'}~C_2'$ where $C_1' \ge C_1$, $C_2' \ge C_2$. There are three cases to consider depending on the reduction $C \rightarrow D$. 
		\begin{itemize}
			\item $C_1$ is $\return(V)$ for some $V$. Then, $C'_1 \rightarrow^* \return(V')$ for some $V' \ge V$. Thus,
			$$C' \rightarrow^* \letin{x = \return(V')}~C_2' \rightarrow C_2'[x := V'] $$ 
			and $C_2'[x := V'] \ge C_2[x := V]$, this concludes this case.
			\item $C_1$ is $\sigma(V;y.B)$. Then, $C'_1 \rightarrow^* \sigma(V';y.B')$ for some $V' \ge V$ and $B' \ge B$. Then,
			$$C' \rightarrow^* \letin{x = \sigma(V';y.B')}~C_2' \rightarrow \sigma(V';y. \letin{x = B'}~C_2')$$ 
			and $\sigma(V';y. \letin{x = B'}~C_2') \ge \sigma(V;y. \letin{x = B}~C_2)$, this concludes this case.
			\item If $C_1 \rightarrow D_1$ for some $D_1$, then we have, by induction hypothesis on $C_1'$ the existence of $D_1'$ such that
			$$C' \rightarrow^+ \letin{x = \mathcal{E}[D_1']}~C_2'$$
			and this computation is greater than $\letin{x = D_1}~C_2$. This concludes this case. 
		\end{itemize}
		\item If $C \equiv \handle{H}~B$. Then $C' \equiv \handle{H'}~B'$. Again, three cases to consider:
		\begin{itemize}
			\item $If B \rightarrow A$, then we can find, by induction hypothesis, $A'$ such that $B' \rightarrow^+ A'$ and $A' \ge A$, and we have $C' \rightarrow^+ \handle{H'}~\mathcal{E}[A']$.
			\item If $B$ is $\return(V)$ for some $V$. Then, $B' \rightarrow^* \return(V')$ for some $V' \ge V$. Thus,
			$$C' \rightarrow^* \handle{H'}~\return(V') \rightarrow C_r'[x := V']$$
			and we have $C_r'[x := V'] \ge C_r[x := V]$, this concludes this case.
			\item If $B$ is $\sigma_i(V;y.A)$. Then, $B' \rightarrow^* \sigma_i(V';y.A')$ for some $V' \ge V$ and $B' \ge B$. Then,
			$$C' \rightarrow^* \handle{H'}~\sigma_i(V';y. A') \rightarrow C_i'[x := V'][r := \la y. \handle{H'}~A']$$  
			and $$C_i'[x := V'][r := \la y. \handle{H'}~A'] \ge C_i[x := V][r := \la y. \handle{H}~A]$$, this concludes this case.
		\end{itemize}
		\item If $C \equiv \case(V;C_1,\dots,C_k)$. Then $C' \equiv \case(V';C_1',\dots,C_k')$. And we can easily concludes because the only $V \ge \underline{n}$ for a given $\underline{n}$ is $V \equiv \underline{n}$. 
	\end{itemize}
\end{proof}

Now that we have this lemma, we can prove that our translation is correct up to approximation.

\begin{lemma}
	If $C \rightarrow D$, then there exists $D'$ such that $\sem{C} \rightarrow^+ D'$ and $D' \ge \sem{D}$. 
\end{lemma}

\begin{proof}
	For this proof, we will consider that all values of type unit are equal. We proceed by induction on $C \rightarrow D$. All cases are straightforward except for the $\case$, that we detail here. 
	\begin{itemize}
		\item Consider the reduction $\case(0;0 \mapsto C;\mathtt{succ}(x) \mapsto D) \rightarrow C$. We have 
		\begin{align*}
			\sem{\case(0;0 \mapsto C;\mathtt{succ}(x) \mapsto D} &\triangleq \letin{a = (\handle{H(C,D)}~\mathtt{Z}_{\sem{\type{Nat}}}~())}~(a~(\la f,g. g~())) \\
			&\rightarrow^+ \letin{a = (\handle{H(C,D)}~\return(()))}~(a~(\la f,g. g~())) \\
			&\rightarrow^+ \letin{a =  \return(\la p. p~\mathtt{Z}_{\type{Unit} \rightarrow \type{Unit}}~(\la x. \sem{C}))}~(a~(\la f,g. g~())) \\ 
			&\rightarrow^+ (\la f,g. g~())~(\mathtt{Z}_{\type{Unit} \rightarrow \type{Unit}})~(\la x. \sem{C}) \\ 
			&\rightarrow^+ \sem{C} 
		\end{align*}
		This concludes this case. 
		\item Consider the reduction $\case(\mathtt{succ}(V);0 \mapsto C;\mathtt{succ}(n) \mapsto D) \rightarrow D[n := V]$.
		We have: 
		\begin{align*}
			&\phantom{\triangleq} \sem{\case(\mathtt{succ}(V);0 \mapsto C;\mathtt{succ}(x) \mapsto D} \\
			&\triangleq \letin{a = (\handle{H(C,D)}~\sem{\mathtt{succ}(V)}~())}~(a~(\la f,g. g~())) \\
			&\rightarrow^+ \letin{a = (\handle{H(C,D)}~\sigma(y. \sem{V}()))}~(a~(\la f,g. g~()))
		\end{align*}
		We can show that the following context:
		$$\mathcal{E} \triangleq \letin{z = \letin{\alpha =\handle{H(C,D)}~[]}~\alpha~(\la f,g. \letin{z = f~()}~\mathtt{Z}_{\sem{T}})}~\return(\mathtt{Z}_{\type{Unit}})$$ is administrative. 
		\begin{itemize}
			\item We have:
			\begin{align*}
				&\phantom{\rightarrow^*}~\handle{H(C,D)}~\sigma(y.B) \\ 
				&\rightarrow^+ \letin{n = \return( \\
				& \la x. \letin{z = (\la y. \handle{H(C,D)}~B)~\underline{0}~(\la f,g. \letin{z = f~()}~\mathtt{Z}_{\sem{T}})}~\return(\mathtt{Z}_{\type{Unit}}))}~...
			\end{align*}
			Also, 
			$$ \letin{z = (\la y. \handle{H(C,D)}~B)~\underline{0}~(\la f,g. \letin{z = f~()}~\mathtt{Z}_{\sem{T}})}~\return(\mathtt{Z}_{\type{Unit}}) \rightarrow^* \mathcal{E}[B]$$ so let us call this computation $B'$ for the sake of clarity. We have:  
			\begin{align*}
				&\phantom{\rightarrow^*}~\handle{H(C,D)}~\sigma(y.B) \\ 
				&\rightarrow^+ \letin{n = \return(\la x. B')}~\return(\la p. p~(\la x. \sigma(z. n~()))~(\la x. \sem{D}))
			\end{align*}
			By going back to the previous computation in the whole context, we obtain:
			\begin{align*}
				& \phantom{\rightarrow^*}~\mathcal{E}[\sigma(V;y.B)] \\
				&\rightarrow^+ \letin{k = \sigma(z. (\la x. B')~())}~\return(\mathtt{Z}_{\type{Unit}}) \\ 
				&\rightarrow^+ \letin{k = \sigma(z. B')}~\return(\mathtt{Z}_{\type{Unit}}) \\
				&\rightarrow^+ \sigma(z.\letin{k = B'}~\return(\mathtt{Z}_{\type{Unit}}))
			\end{align*}
			And $\letin{k = C}~\return(())$ with $k$ of type unit is $\eta$-equivalent to $C$. This concludes this case.
			\item  For the case of $\return$, we have:
			\begin{align*}
				& \phantom{\rightarrow^*}~\mathcal{E}[\return(())] \\
				&\rightarrow^+ \letin{z = \letin{a = \handle{H(C,D)}~\return(())} \\
				&(a~(\la f,g. \letin{z = f~()}~\mathtt{Z}_{\sem{T}}))}~\return(\mathtt{Z}_{\type{Unit}}) \\
				&\rightarrow^+ \return(()) 
			\end{align*}  
			And this concludes this case. 
		\end{itemize}
		In particular, we have:
		$$\mathcal{E}[\sem{V}~()] \ge \sem{V}() $$
		for any $\Gamma \p V : \type{Nat}$ in PCF. 
		This show in particular that we have: 
		\begin{align*}
			&\phantom{equiv}\sem{\case(\mathtt{succ}(V);0 \mapsto C;\mathtt{succ}(x) \mapsto D} \\
			&\triangleq \letin{a = (\handle{H(C,D)}~\sem{\mathtt{succ}(V)}~())}~(a~(\la f,g. g~())) \\
			&\rightarrow^+ \letin{a = (\handle{H(C,D)}~\sigma(y. \sem{V}()))}~(a~(\la f,g. g~())) \\ 
			&\rightarrow^+ \sem{D}[n := V']
		\end{align*}
		with $V' \ge \la x. \sem{V}~()$, which is $\eta$-equivalent with $\sem{V}$.  
		Thus, we have indeed that 
		$$\sem{\case(\mathtt{succ}(V);0 \mapsto C;\mathtt{succ}(x) \mapsto D} \rightarrow^+ \sem{D}[n := V']$$
		 with $\sem{D}[n := V'] \ge \sem{D}[n := \sem{V}]$.  
	\end{itemize}
\end{proof}

\begin{lemma}
	If $\Gamma \p C : T$, $C \rightarrow D$ and $A \ge \sem{C}$, then there exists $B$ such that $A \rightarrow^+ B$ with $B \ge \sem{D}$ 
\end{lemma}

\begin{proof}
	This is a consequence of the two previous lemmas. 
\end{proof}

From this simulation, we obtain:
\else 
\begin{lemma}
	If $\Gamma \p C : T$, $C \rightarrow D$ and $A \ge \sem{C}$, then there exists $B$ such that $A \rightarrow^+ B$ with $B \ge \sem{D}$ 
\end{lemma}

The details of this proof can be found in the supplementary material, but the result is similar to the one obtained in \cite{Hillerstrom2018:Shallow}, Theorem~9. From this simulation, we obtain:
\fi 

\begin{theorem}
	For any term $\p C : \type{Nat}$ in PCF, there is a term of finitary \HEPCF\ $C_f \triangleq \letin{a = \sem{C}}~(a~())$ such that $\p_E C_f: \type{Unit}$ and $ET(C_f)$ represents the normal form of $C$ (if it exists, otherwise the $ET(C_f) \equiv \bot$).
\end{theorem}

\begin{proof}
	Suppose that $\p C : \type{Nat}$ has no normal form, then $\sem{C}$ has no normal form, and thus the effect tree of $C_f$ is $\bot$. Otherwise, suppose that $C \rightarrow^* \return(V)$ with $V$ a closed value of type $\type{Nat}$. Then, $\sem{C} \rightarrow^* D$ with $D \ge \sem{\return(V)}$. In particular, $D \rightarrow^* \return(V')$ with $V' \ge \sem{V}$. Thus, the effect tree of $C_f$ is the same as $\sem{V}$. Because $V$ is a closed value of type $\type{Nat}$, this effect tree is the tree representation of the integer $V$, this concludes the proof. 
\end{proof}

\begin{corollary}
	Effect trees produced by \HEPCF\ terms are not \MSO-decidable.  
\end{corollary}

Indeed, given a \PCF\ term  of type $\type{Nat}$, we can translate it into 
\HEPCF\, and ask if the effect tree of this \HEPCF\ term contains a $\bot$, and 
this would decide if the initial term halts. Thus, this shows that the power of 
handlers are simply too expressive for \MSO\ model-checking.
\ifLong

\begin{remark}
	Our proof heavily relies on the fact that handlers can change the output 
	type, which allows here to do this encoding of pairs of computation 
	internally and returns at the end only the second computation. A careful 
	reader might wonder if the control of this output type is necessary for 
	undecidability, and that we could have a decidable \MSO\ model-checking for 
	restrained handlers, with a fixed output type, for example $\type{Unit} 
	_E\Rightarrow_{E'} \type{Unit}$. We show in the next section that this 
	restriction still leads to a language for which \MSO\ model-checking is 
	undecidable.
\end{remark}
\section{Handlers in ECPS}

In this section, we present a language written directly in CPS-style, with handlers, called HECPS. This language is based on the language described in \cite{Matache2018:MasterThesis,MatacheStaton2019:LogicAlgebraicEffects}. The handlers in ECPS seems less expressive than the handlers available in \HEPCF\, especially because we cannot chose the output and input type, however we will show that we still have an undecidability result for HECPS.

Formally, we define HECPS as the language with the following types and grammar: 
\begin{align*}
	T &::= \type{B} \mid \neg_{E}(T_1,\dots,T_n) \mid \type{k} \\
	E &::= \emptyset \midd \set{\sigma : \type{B} \ararrow k} \cup E \\
	V,W &::= v \midd \underline{n} \midd x \midd \la (x_1,\dots,x_n). t \midd \fix~f. V \\
	t,u &::= V~(W_1,\dots,W_n) \midd \sigma(V; x.t) \midd \return \midd \handle{H}~t \midd \case(V;t_1,\dots,t_k) \\ 
	H &::= \set{\return \mapsto t_r} \cup \set{\sigma_i(x;r) \mapsto t_i \mid 1 \le i \le n}
\end{align*}
The intuition behind this language is that all terms $t,u$ must be of type $o$, where $o$ is intuitively a return type corresponding to the type of tree in the $\la$-Y calculus. In particular, a value of HECPS is either of base type, or a function that outputs in type $o$, described by the type $\neg_{E}(T_1,\dots,T_n)$ meaning intuitively $T_1 \rightarrow_E \cdots \rightarrow_E T_n \rightarrow_E o$. Then,
the obvious difference with \HEPCF\ is the absence of $\letin{\cdot}$ expressions, because it is captured by a use of continuation, and there
is no way to return a value, because we can only stop the computation with a $\return$ of type $o$. 

The typing rules that differs from \HEPCF\ are given in Figure~\ref{f:HECPStyping}.

\begin{figure}
	\begin{framed}
		\begin{center}
			\AXC{$\Gamma, x_1 : T_1, \dots, x_n : T_n \p t $}
			\UIC{$\Gamma \p \la (x_1,\dots,x_n).t: \neg(T_1,\dots,T_n)$}
			\DP 
			\qquad 
			\AXC{$\Gamma, f : \neg(T_1,\dots,T_n) \p V : \neg(T_1,\dots,T_n)$}
			\UIC{$\Gamma \p \fix~f.V : \neg(T_1,\dots,T_n)$}
			\DP 
			\\
			\vvskip  
			\AXC{$\Gamma \p V : \neg(T_1,\dots,T_n)$}
			\AXC{$(\Gamma \p W_i : T_i)_{1 \le i \le n}$} 
			\BIC{$\Gamma \p V~(W_1,\dots,W_n) $}
			\DP 
			\qquad 
			\AXC{}
			\UIC{$\Gamma \p \return $}
			\DP 
			\\ 
			\vvskip  
			\AXC{$(\sigma : \type{B} \ararrow k) \in \Sigma$}
			\AXC{$\Gamma \p V : \type{B}$}
			\AXC{$\Gamma, x : \type{k} \p t$}
			\TIC{$\Gamma \p \sigma(V;x.t)$}
			\DP 
			\\ \vvskip 
			\AXC{$\Gamma \p V : \type{k}$}
			\AXC{$(\Gamma \p t_i)_{1 \le i \le k}$}
			\BIC{$\Gamma \p \case(V;t_1,\dots,t_k)$}
			\DP 
			\qquad 
			\AXC{$\Gamma \p H : E' \Rightarrow E$}
			\AXC{$\Gamma \p_{E'} t$}
			\BIC{$\Gamma \p_E \handle{H}~t$}
			\DP  
			\\ 
			\vvskip  
			\AXC{$E' = \set{\sigma_i : \type{B}_i \ararrow k_i \mid 1 \le i \le n}$}
			\AXC{$\Gamma \p_E t_r $}
			\AXC{$(\Gamma, x : \type{B}_i, r : \neg_E(k) \p_E t_i)_{1 \le i \le n}$}
			\TIC{$\Gamma \p \set{\return \mapsto t_r} \cup \set{\sigma_i(x;r) \mapsto t_i \mid 1 \le i \le n} : E' \Rightarrow E$}
			\DP 
		\end{center}
	\end{framed}
	\caption{Static Semantics of HECPS}
	\label{f:HECPStyping}
\end{figure}

As for the semantics, it is similar to the one described in Figure~\ref{f:HEPCFSemantics}. The important difference with what we saw before is that we cannot control the input or output type, we need to work with type $o$. Because of this, the previous translation from PCF, that relied strongly on the output type, does not work directly. However, we will again prove that we can encode PCF in this language, and as before we rely on the idea to encode shallow handlers as deep handlers. Formally, a shallow handler is syntactically similar to a deep handler, but the typing rule and semantics are different, as described in Figure~\ref{f:shallow}. We differentiate shallow handlers from usual handlers by adding a $\dagger$ in the handle operation. In the typing derivation the difference is subtle, only the set of effect of the continuation $r$ changes. However, in the semantics the difference is flagrant, the continuation is not called with the handler again. 

\begin{figure}
	\begin{framed}
		\begin{center}
			\AXC{$\Gamma \p H : E' \Rightarrow^{\dagger} E$}
			\AXC{$\Gamma \p_{E'} t$}
			\BIC{$\Gamma \p_E \shandle{H}~t$}
			\DP  
			\\ 
			\vvskip  
			\AXC{$E' = \set{\sigma_i : \type{B}_i \ararrow k_i \mid 1 \le i \le n}$}
			\AXC{$\Gamma \p_E t_r $}
			\AXC{$(\Gamma, x : \type{B}_i, r : \neg_{E'}(k) \p_E t_i)_{1 \le i \le n}$}
			\TIC{$\Gamma \p \set{\return \mapsto t_r} \cup \set{\sigma_i(x;r) \mapsto t_i \mid 1 \le i \le n} : E' \Rightarrow^{\dagger} E$}
			\DP
			\\
			\vvskip  
			\AXC{$t \rightarrow t'$} 
			\UIC{$\shandle{H}~t \rightarrow \shandle{H}~t' $}
			\DP 
			\\ 
			\vvskip 
			For the following rules, we define $H =\set{\return \mapsto t_r} \cup \set{\sigma_i(x;r) \mapsto t_i \mid 1 \le i \le n}$ 
			\\
			\vvskip 
			\AXC{}
			\UIC{$\shandle{H}~\return \rightarrow t_r$} 
			\DP 
			\\
			\vvskip 
			\AXC{}
			\UIC{$\shandle{H}~\sigma_i(V;x.t) \rightarrow t_i[x := V, r := (\la x. t)] $}
			\DP 
		\end{center}
	\end{framed}
	\caption{Shallow Handlers in ECPS}
	\label{f:shallow}
\end{figure}

\subsection{ECPS with Shallow Handlers}

We introduce shallow handlers in order to decompose the translation from PCF, we will first show that we can encode PCF in ECPS with shallow handlers, and then we will see how we can encode shallow handlers using deep handlers in ECPS (the important point being that the proof given in the previous section and in \cite{Hillerstrom2018:Shallow} does not work anymore).

\begin{definition}[Translation from PCF to ECPS with Shallow Handlers]
	As in the previous section, we need to encode the infinite type of natural numbers into finite types. We introduce the effect $E = \set{\sigma : \type{Unit} \ararrow 1}$. We define the \emph{zero} term, denoted $\mathtt{Z}_o$ by
	$ \mathtt{Z}_o = \return$. 
	
	The translation $\sem{\cdot}$ from PCF to ECPS with shallow handlers is given by
	$$ \sem{\type{Nat}} = \neg_E(1) \qquad \sem{T \rightarrow U} = \neg_E(\sem{T},\neg_{E}(\sem{U})).$$
	And  we want, as usual for a CPS translation, that if $\Gamma \p V : T$ then $\sem{\Gamma} \p \sem{V} : \sem{T}$ and that 
	if $\Gamma \p C : T$ then $\sem{\Gamma} \p \sem{C} : \neg_{E}(\neg_{E}(\sem{T}))$
	\begin{align*}
		\sem{0} &\triangleq \la x. \mathtt{Z}_o 
		&\sem{\mathtt{succ}(V)} &\triangleq \la x. \sigma(y. \sem{V}~y) \\
		\sem{x} &\triangleq x 
		&\sem{\la x. C} &\triangleq \la (x,k). \sem{C}~k \\ 
		\sem{\fix f. V} &\triangleq \fix f. \sem{V}  
		&\sem{V~W} &\triangleq \la k. \sem{V}~(\sem{W},k) \\ 
		\sem{\return(V)} &\triangleq \la k. k~\sem{V} 
		&\sem{\letin{x = C}~D} &\triangleq \la k. \sem{C} (\la x. \sem{D}~k)  
	\end{align*}   
	$$\sem{\case(V;0 \mapsto C;\mathtt{succ}(n) \mapsto D)} \triangleq \la k. \shandle{H(C,D,k)}~\sem{V} \underline{0}$$
	where $H(C,D,k)$ is defined by, when $C$ and $D$ have type $T$, and $k$ of type $\neg_{E}(T)$ by: 
	\begin{align*}
		\return &\mapsto \sem{C}~k &\sigma(n) &\mapsto \sem{D}~k
	\end{align*}
	This time, the encoding of case is rather simple. Indeed, we will show that shallow handlers naturally encodes the conditional with this construction.  
\end{definition}

\begin{lemma}
	If $\Gamma \p V : T$ then $\sem{\Gamma} \p \sem{V} : \sem{T}$ and if $\Gamma \p C : T$ then $\sem{\Gamma} \p \sem{C} : \neg_{E}(\neg_{E}(\sem{T}))$
\end{lemma}

\begin{proof}
	The proof is straighforard. 
\end{proof}

Now, we want to prove that the translation is a correct simulation.

\begin{lemma}
	If $\Gamma \p C : T$ and $C \rightarrow D$ then for all $\sem{\Gamma} \p k : \neg_{E}(\sem{T})$, $\sem{C}~k \rightarrow^+ \sem{D}~k$.  
\end{lemma}

\begin{proof}
	The only difference with a standard CPS translation is for the $\case$ constructor. Thus, only detail this case. 
	
	\begin{itemize}
		\item Consider the reduction $\case(0,0 \mapsto C; \mathtt{succ}(n) \mapsto D) \rightarrow C$.
		Then, we have:
		\begin{align*}
			\sem{\case(0,0 \mapsto C; \mathtt{succ}(n) \mapsto D)}~k &\rightarrow \shandle{H(C,D,k)}~\sem{0}() \\
			&\rightarrow  \shandle{H(C,D,k)}~\return \\ 
			&\rightarrow  \sem{C}~k 
		\end{align*}
		This concludes this case. 
		\item Consider the reduction $\case(\mathtt{succ}(V),0 \mapsto C; \mathtt{succ}(n) \mapsto D) \rightarrow D[n := V]$	
		Then, we have:
		\begin{align*}
			\sem{\case(\mathtt{succ}(V),0 \mapsto C; \mathtt{succ}(n) \mapsto D)}~k &\rightarrow \shandle{H(C,D,k)}~\sem{\mathtt{succ}(V)}() \\
			&\rightarrow  \shandle{H(C,D,k)}~\sigma(y. \sem{V}~y) \\ 
			&\rightarrow  \sem{D}[n := \la y. \sem{V}~y]~k \\ 
			&\equiv_{\eta} \sem{D}[n := \sem{V}]~k
		\end{align*}
		We need to use the $\eta$ equivalence to get exactly the result, but we could always formally avoid it, but for the sake of simplicity we just consider that we work with $\eta$-equivalence. 
	\end{itemize}
\end{proof}

From this, we can conclude the undecidability result for ECPS with shallow handlers.

\begin{lemma}{}
	For any term $\p C : \type{Nat}$ in PCF, there is a term of ECPS with shallow handlers $t_f \triangleq \sem{C}~(\la n. n~\underline{0})$ such that $\p_E t_f$ such that the effect tree of $t_f$ represents the normal form of $C$ (if it exists, otherwise the effect tree is $\bot$).
\end{lemma}

\begin{proof}
	Suppose that  $\p C : \type{Nat}$ has no normal form, then $\sem{C}$ has no normal form, and thus the effect tree of $t_f$ is $\bot$. Otherwise, $C \rightarrow^* \return(V)$ with $V$ a value of type $\type{Nat}$. Then, with the previous lemma, we obtain 
	$$t_f \rightarrow^* \sem{\return(V)}~(\la n. n~\underline{0}) \rightarrow^* \sem{V}~\underline{0}  $$
	and we can easily see that the effect tree of $\sem{V}~\underline{0}$ represents exaclty the integer $V$. 
\end{proof}

Thus, to conclude on the undecidability of HECPS, we need to show that we can simulate shallow handlers using deep handlers.

\subsection{Encoding Shallow Handlers with Deep Handlers}

As we are only interested in encoding the previous translation, we fix the set of effect $E = \set{\sigma : \type{Unit} \ararrow 1}$. We introduce a duplication of this effect, and consider the set $E_d = \set{\sigma_d : \type{Unit} \ararrow 1}$. We will use this operation $\sigma_d$ as a 'marked' effect, similarly to the use of 'marked' letters that we can find in some Turing Machine problems. In the following proof, we will mainly work directly on effect trees, because we will introduce handlers and each of these handlers will slightly transform 
the computation by just changing a bit the effect tree. We believe this proof is far more intuitive. However, if we want to prove formally a simulation, we would need to use administrative context like before. 

As a first step toward this intuition of marked effects, we introduce the marking handlers, denoted $H_d$, and described by:
\begin{align*}
	H_d \triangleq \{ \return \mapsto \return;
	\qquad \sigma(r) \mapsto \sigma_d(x. r~x); 	
	\qquad \sigma_d(r) \mapsto \sigma_d(x. r~x) \}
\end{align*}
It is easy to see that the type of this handler is $(E \cup E_d) \Rightarrow E_d$ (it can also be the weaker type $(E \cup E_d) \Rightarrow (E \cup E_d)$ , and that this handler takes a computation $C$ and simulates $C$ but replace all effects $\sigma(x.t)$ by $\sigma_d(x.t)$, without changing anything else. 
We now introduce a handler that change everything except the root of the tree, to do this we use the previous handler:
\begin{align*}
	H_r \triangleq \{ &\return \mapsto \return; \qquad \sigma(r) \mapsto \sigma(y. \handle{H_d}~(r~y))	
\end{align*}
This handler $H_r$ has type $E \Rightarrow (E \cup E_d)$. We have, 
\begin{itemize}
	\item If $\Gamma \p_{E} t$, if $t \rightarrow t'$ then $\handle{H_r}~t \rightarrow \handle{H_r}~t'$.
	\item $\handle{H_r}~\return \rightarrow \return$
	\item $\handle{H_r}~\sigma(x.t) \rightarrow \sigma(y. \handle{H_d}~(\la x.(\handle{H_r}~t)~y))$
	Using $\alpha$-equivalence, we obtain:
	$\handle{H_r}~\sigma(x.t) \rightarrow^+ \sigma(x. \handle{H_d}~(\handle{H_r}~t))$. 
\end{itemize} 
For the sake of clarity, let us note $\handle{ H_d \circ H_r}~t$ for $\handle{H_d}~(\handle{H_r}~t)$. 
We have:
\begin{itemize}
	\item If $\Gamma \p_{E \cup E_d} t$, if $t \rightarrow t'$ then $\handle{H_d \circ H_r}~t \rightarrow^+ \handle{H_d \circ H_r}~t'$.
	\item $\handle{H_d \circ H_r}~\return \rightarrow^+ \return$.
	\item $\handle{H_d \circ H_r}~\sigma(x.t) \rightarrow^+ \sigma_d(x. \handle{H_d}~(\handle{H_d \circ H_r}~t))$. 
\end{itemize}
Now remark that applying $H_d$ twice does not change anything, by definition of $H_d$. And as we can see with the previous description, the effect tree of $\handle{H_d \circ H_r}~t$ is the same as $\handle{H_d}~t$ alone. Thus, we can show that the effect tree of $\handle{H_r}~t$ is the effect tree of $t$ when we replace all $\sigma$ by $\sigma_d$ except on the root. Thus, we have just found a way to distinguish the root of the effect tree from the rest, and we can use that to simulate a shallow handler.

Suppose that we want to simulate the following shallow handler $h$ with a type corresponding to the shallow handlers we used in the previous translation:
\begin{prooftree}
	\AXC{$\Gamma \p_E t_r $}
	\AXC{$\Gamma, r : \neg_{E}(1) \p_E t_{\sigma}$}
	\BIC{$\Gamma \p \set{\return \mapsto t_r} \cup \set{\sigma(r) \mapsto t_{\sigma}} : E \Rightarrow^{\dagger} E$}
\end{prooftree} 
Consider the following (deep) handler:
\begin{align*}
	H_h \triangleq \{ \return \mapsto t_r;
	\qquad \sigma(r) \mapsto t_{\sigma};
	\qquad \sigma_d(r) \mapsto \sigma(x. r~x) \}
\end{align*}
that has the type $(E \cup E_d) \Rightarrow E$
and consider the following term: $\handle{H_h \circ H_r}~t$. 
We have: 
\begin{itemize}
	\item If $t \rightarrow t'$ then $\handle{H_h \circ H_r}~t \rightarrow^+ \handle{H_h \circ H_r}~t'$
	\item $\handle{H_h \circ H_r}~\return \rightarrow^+ t_r$
	\item $\handle{H_h \circ H_r}~\sigma(y.t) \rightarrow^+ t_\sigma[r := (\la x. \handle{H_h \circ H_d \circ H_r}~t))]$
\end{itemize}
And now we have to show that $\handle{H_h \circ H_d \circ H_r}~t$ behaves like $t$ (the context is administrative). We already know that $H_d \circ H_r$ behave like $H_d$, and on effect tree with effect only in $E_d$, the handler $H_h$ will just replace every $\sigma_d$ by $\sigma$. Thus, the effect tree of $t$ is exactly the effect tree of $\handle{H_h \circ H_d \circ H_r}~t$. Thus, $H_h \circ H_r$ behaves like a shallow handler. We can then take the previous encoding of $PCF$ in ECPS with shallow handlers, and replace the use of a shallow handler $h$ by the use of $H_h \circ H_r$, and we obtain a valid encoding of $PCF$ to ECPS with deep handlers. 

\begin{theorem}
	HECPS is not \MSO-decidable
\end{theorem}

\section{EPCF with Restricted Handlers (Generic Effect)}
\else 
\begin{remark}
	This proof is dependent on the fact that handlers can change the output type, which allows here to do this encoding of pairs of computation internally, create a copy and returns at the end only the second computation. A careful reader might wonder if the control of this output type is necessary for undecidability, and that we could have a decidable \MSO\ model-checking for restrained handlers, with a fixed output type, for example $\type{Unit} _E\Rightarrow_{E'} \type{Unit}$. This is typically the kind of handler we would obtain with an extension of the \ECPS\ language \cite{MatacheStaton2019:LogicAlgebraicEffects,Matache2018:MasterThesis} with handlers. However, we show in the supplementary material that this restriction still leads to a language for which \MSO\ model-checking is undecidable.
\end{remark}
\subsection{Recovering Decidability in a Calculus for Generic 
Effects}
\fi
Does the results in the last section mean that we can have algebraic effects 
but that we have to get away with handlers if HOMC is our concern? Essentially, 
it rather shows that general handlers, define a transformation on effect trees 
that is simply \emph{too expressive} for HOMC. However, this does not mean that 
\emph{simpler} tree transformations should not be expressible in HORS, and this 
section is devoted to introducing a class of handlers whose underlying tree 
transformations is amenable to HOMC. Intuitively, an interpretation of an 
algebraic operation $\sigma : \type{B} \ararrow k$ could be just a function of 
type $B \rightarrow k$, with no access to the continuation. This is of course 
much less expressive than standard handlers, but we will argue that what we 
obtain is not \emph{too} restrictive. Formally, we could capture this by a 
restriction of handlers, defined in the following language, called \GEPCF\ (for 
Generic Effects \PCF):
\begin{align*}
&\mbox{(Types)}&	T,U &::= \type{B} \mid T \rightarrow_E U \mid k \\
&\mbox{(Effects)}&	E &::= \emptyset \midd \set{\sigma : \type{B} \ararrow k} \cup E \\
&\mbox{(Values)}&	V,W &::= v \midd \underline{n} \midd x \midd \la x. C \midd \fix~f. V \\
&\mbox{(Computations)}&	C,D &::= V~W \midd \sigma(V; x.C) \midd \return(V) \midd \letin{x = C}~D \\
&&&\phantom{::=} \midd  \handle{H}~C \midd \case(V;C_1,\dots,C_k) \\ 
&\mbox{(Handlers)}&	H &::= \set{\return(x) \mapsto C_r} \cup \set{\sigma_i(x) \mapsto C_i \mid 1 \le i \le n}
\end{align*}
Note that as expected, the handler for an algebraic operation does not have 
access to the continuation. The static and dynamic semantics is in 
Figure~\ref{f:GEPCFsemantics}.

\begin{figure}
	\begin{framed}
		\begin{center}
			\AXC{$\Gamma \p H : T _E\Rightarrow_{E'} U $}
			\AXC{$\Gamma \p_{E} C : T$}
			\BIC{$\Gamma \p_{E'} \handle{H}~C : U$}
			\DP  
			\\ 
			\vvskip  
			\AXC{$E = \set{\sigma_i : \type{B}_i \ararrow k_i \mid 1 \le i \le n}$}
			\AXC{$\Gamma, x : T \p_{E'} C_r : U $}
			\AXC{$(\Gamma, x : \type{B}_i, \p_{E'} C_i : k_i)_{1 \le i \le n}$}
			\TIC{$\Gamma \p \set{\return(x) \mapsto C_r} \cup \set{\sigma_i(x) \mapsto C_i \mid 1 \le i \le n} : T _E\Rightarrow_{E'} U$}
			\DP
			\\
			\vvskip  
			\AXC{$t \rightarrow t'$} 
			\UIC{$\handle{H}~t \rightarrow \handle{H}~t' $}
			\DP 
			\\ 
			\vvskip 
			For the following rules, we define $H =\set{\return(x) \mapsto C_r} \cup \set{\sigma_i(x) \mapsto C_i \mid 1 \le i \le n}$ 
			\\
			\vvskip 
			\AXC{}
			\UIC{$\handle{H}~\return(V) \rightarrow C_r[x := V]$} 
			\DP 
			\\
			\vvskip 
			\AXC{}
			\UIC{$\handle{H}~\sigma_i(V;y.C) \rightarrow \letin{y = C_i[x := V]}~(\handle{H}~C) $}
			\DP 
		\end{center}
	\end{framed}
	\caption{Interpretations in GPCF}
	\label{f:GEPCFsemantics}
\end{figure}

Please observe that we e can always simulate a \GEPCF\ handler by a \HEPCF\ 
handler:
$$ \return \mapsto C_r \qquad \sigma_i(x;r) \mapsto \letin{z = C_i}~(r~z) $$
Thus, those restrained handlers are a specific case of standard handlers in 
which the handling of an operation $\sigma_i$ always has the shape $\letin{z = 
C_i}~(r~z)$. Informally, if we look at it with the point of view of tree 
transformations, this is exactly what we would obtain from a generic effect for 
the effect tree monad, that is why we call this language \GEPCF. 

Let us now show that this language can be translated into the $\laY$-calculus, 
thus proving the \MSO-decidability.

\begin{definition}[Translation from \GEPCF\ to $\la$-Y calculus]
	We define the translation $\toY{\cdot}$ on types and terms. For the sake of conciseness, in types, we use some abuse of notations, such as pairs in the left-hand side of an arrow type to denote the obvious curried version of this type. 
	$$ \toY{\type{B}} \triangleq o \qquad \toY{k} \triangleq o^k \rightarrow o \qquad \toY{(T \rightarrow_E U)} \triangleq \toY{T} \rightarrow \toY{E} \rightarrow \neg\toY{U}  \rightarrow o $$
	$$ \toY{\emptyset} \triangleq () \qquad \toY{E \cup \sigma : \type{B} \ararrow k} \triangleq (\toY{E},o \rightarrow \neg\toY{k}  \rightarrow o)  $$
	And we want to show that if $\Gamma \p V : T$ then $\toY{\Gamma} \p \toY{V} : \toY{T}$ and that if $\Gamma \p_E C : T$ then $\toY{\Gamma} \p \toY{C} : \toY{E} \rightarrow (\toY{T} \rightarrow o) \rightarrow o$
	Informally, we will do a CPS-translation in which there is also a \emph{handler continuation}, of type $\toY{E}$, giving an interpretation for all algebraic effects in $E$. We take a $\la Y$-calculus with a signature $\Sigma$  with at least all the constants values of EPCF. 
	\begin{align*}
		\toY{v} &\triangleq v 
		&\toY{\underline{n}} &\triangleq \la (x_1,\dots,x_k). x_n \\
		\toY{x} &\triangleq x 
		&\toY{(\la x. C)} &\triangleq \la x. \toY{C} \\
		\toY{(\fix~x.V)} &\triangleq Y~(\la x. \toY{V}) 
		&\toY{(V~W)} &\triangleq \toY{V}~\toY{W} \\ 
		\toY{(\return(V))} &\triangleq \la (\vect{h},c). c~\toY{V} 
		&\toY{(\sigma_i(V;x.C))} &\triangleq \la (\vect{h},c). h_i~\toY{V}~(\la x. \toY{C}~\vect{h}~c)
	\end{align*}
	\begin{align*}
		\toY{(\letin{x = C}~D)} &\triangleq \la (\vect{h},c). \toY{C}~\vect{h}~(\la x. \toY{D}~\vect{h}~c) \\
		\toY{(\case(V;C_1,\dots,C_k))} &\triangleq \la (\vect{h},c). \toY{V}~(\toY{C_1}~\vect{h}~c)~\cdots~(\toY{C_k}~\vect{h}~c) \\
		\text{We suppose } H &= \set{\return(x) \mapsto C_r ; \sigma_i(x) \mapsto C_i \mid 1 \le i \le n} \\
		\toY{(\handle{H}~C)} &\triangleq \la (\vect{h},c). \toY{C}~(\la (x,r). \toY{C_1}~\vect{h}~r)~\cdots~(\la (x,r). \toY{C_n}~\vect{h}~r)~(\la x. \toY{C_r}~\vect{h}~c).  
	\end{align*} 
\end{definition}

\begin{lemma}
	If $\Gamma \p V : T$ then $\toY{\Gamma} \p \toY{V} : \toY{T}$ and that if $\Gamma \p_E C : T$ then $\toY{\Gamma} \p \toY{C} : \toY{E} \rightarrow \neg \neg \toY{T} $
\end{lemma}

The proof is straightforward. 

\begin{lemma}
	If $\Gamma \p_E C : T$ and $C \rightarrow D$ then for any well-typed continuations $(\vect{h},c)$, we have $\toY{C}_{\vect{h},c} \rightarrow^+ B =_{adm} \toY{D}_{\vect{h},c}$, where $=_{adm}$ is the equivalence relation with $(\la ({h},c). C)~\vect{h}~c \equiv C$.  
\end{lemma}

\begin{proof}
	The proof goes by induction on the derivation of $C \rightarrow D$. Some 
	cases are equivalent to the usual CPS translation so we can safely ignore 
	them. Similarly to Lemma~\ref{l:simulationEPCF}, we have to take care of 
	some administrative reductions. The new cases correspond to the cases of 
	handlers. The contextual case is straightforward since we can reduce the 
	head of the $\la-Y$-term. The most interesting case is the handlers against 
	an algebraic operation. For this case, we have 
	$$\toY{(\handle{H}~\sigma_i(V;y.C))}~\vect{h}~c \rightarrow^+ (\toY{C_i})[x := \toY{V}]~\vect{h}~(\la y. \toY{C}~\vect{\mathcal{H}}~\mathcal{C})$$
	where $\vect{\mathcal{H}}$ and $\mathcal{C}$ denotes the continuations such that 
	$$\toY{(\handle{H}~C)} \triangleq \la (\vect{h},c). \toY{C}~\vect{\mathcal{H}}~\mathcal{C}$$
	On the other hand, we have:
	$$\toY{(\letin{y = C_i[x := V]}~(\handle{H}~C))}~\vect{h}~c \rightarrow_{adm} \toY{C_i}[x := V]~\vect{h}~(\la y. (\la (\vect{h},c). \toY{C} \vect{\mathcal{H}}~\mathcal{C})~\vect{h}~c) $$
	And those two terms are administratively equivalent. 
	This concludes this proof. 
\end{proof}

We can now define a canonical continuation handler for a given set of effects.

\begin{definition}
	Consider a set $E =\set{\sigma_i : \type{B}_i \ararrow k_i}$ of effects. We pose $\Sigma_E = \set{\sigma_i : k_i +1}$ a signature for terms of the $\la$Y-calculus. We define the identity continuation handler for $E$ as the $\la$Y-term of type $\toY{E}$:
	$$(\vect{h}_E)_i \triangleq \la (x,f). \sigma_i~x~(f~(\la x_1,\dots,x_k. x_1))~\cdots~(f~(\la x_1,\dots,x_k. x_k)) $$
	In particular, remark that $\toY{(\sigma_i(V;x.C))}_{\vect{h}_E,c}$ becomes:
	$$\toY{(\sigma_i(V;x.C))}_{\vect{h}_E,c} \rightarrow \sigma_i~\toY{V}~(\toY{C}[x := \toY{\underline{0}}]~\vect{h}_E~c)~\cdots~(\toY{C}[x := \toY{\underline{k-1}}]\vect{h}_E~c) $$ 
	which corresponds exactly to the effect tree of $\sigma_i(V;x.C)$. 
\end{definition}

With this remark, we can conclude with the following theorem:

\begin{theorem}
	For any $\p_E C :T$, for any continuation $\p c : \toY{T} \rightarrow o$, we have:
	$$ET(C)[\return(V) \leftarrow BT(c~\toY{V})] = BT(\toY{V}~\vect{h}_E~c)  $$
\end{theorem}

\begin{proof}
	We can proceed by induction on the effect tree of $C$. 
	\begin{itemize}
		\item If $C$ can be reduced infinitely, then by simulation, $\toY{C}~\vect{h}_E~c$ can be reduced infinitely too. 
		\item If $C \rightarrow^* \return(V)$. Then $\toY{C}~\vect{h}_E~c \rightarrow^* c~\toY{V}$ by definition of $\toY{\return(V)}$.  
		\item If $C \rightarrow^* \sigma(V;x.D)$. Then $\toY{C}~\vect{h}_E~c \rightarrow^* \sigma_i~\toY{V}~(\toY{D}[x := \toY{\underline{0}}]~\vect{h}_E~c)~\cdots~(\toY{D}[x := \toY{\underline{k-1}}]\vect{h}_E~c)$. And we can conclude by induction hypothesis. 
	\end{itemize}
\end{proof}

In particular, if we have $T$ has a base type and a continuation that is the encoding of the identity, we can recover exactly the initial effect tree. Thus, we have:

\begin{corollary}
	For any \GEPCF\ computation $\p C : T$ and $\p c : \toY{T} \rightarrow o$, the \MSO\ model-checking on $ET(C)[\return(V) \leftarrow BT(c~V^*)]$ is decidable.
\end{corollary}

Similarly to what happened in Corollary~\ref{t:EPCFdecidability}, there are of 
course some degrees of freedom in the choice of the continuation $c$.

\section{Related Work}
\label{s:relatedwork}
\paragraph{Effectful Higher-Order Programs and their Verification}
This is definitely not the first paper concerned with the 
verification of 
higher-order effectful programs. the denotational semantics of 
calculi having this 
nature has been studied since Moggi’s seminal work on 
monads~\cite{Moggi1998:ComputationalMonads}, 
thus implicitly providing notions of equivalence and refinement. All this has been 
given a more operational flavor in Plotkin and Power's account on adequacy for 
algebraic effects~\cite{PlotkinPower2003:AlgebraicOperations}, from which the operational semantics presented in 
this paper is greatly inspired. Logics for algebraic effects have been 
introduced by Pretnar and Plotkin~\cite{PlotkinPretnar2008:LogicAlgebraicEffects}, by Matache and
Staton \cite{MatacheStaton2019:LogicAlgebraicEffects}, and by Simpson and 
Voorneveld~\cite{Simpson2019:BehaviouralEquivalenceEffects}. The latter has certainly been another major source of 
inspiration although, as explicitly stated by the 
authors~\footnote{In 
\cite{Simpson2019:BehaviouralEquivalenceEffects}, Section~10, 
paragraph~8: ``We view the infinitary propositional logic of 
this paper as providing a low-level language, into which 
practical high-level finitary logics for expressing program 
properties can potentially be compiled. [...] We view
the development of such high-level logics and their compositional reasoning principles, aimed at
practical specification and verification, as a particularly 
promising topic for future research.''}, 
automatic verification techniques are to be considered out of the scope. In 
fact, we are not aware of any attempt to study the decidability of the 
aforementioned theories.

\paragraph{Higher-Order Model Checking}
Model checking of higher-order programs has been an active topic of 
investigation in the last twenty years, with many positive results, staring 
from the pioneering and partial results by Knapik et al.\cite{Knapik2001:MSOforHORS,Knapi2002:HORSareeasy} to the Ong's 
already mentioned breakthrough result~\cite{Ong2006:HOMC}, followed by Kobayashi and 
co-authors' work on model checking as (intersection) type checking~\cite{KobayashiOng2009:TypeSystemHOMC}.
Noticeably, some of these works go in the direction of extending the 
aforementioned decidability results to higher-order calculi endowed with some 
\emph{specific} form of effect, like probabilistic 
choice~\cite{Kobayashi2020:Termination} or 
some form of resource usage~\cite{Kobayashi2009:TypesHOMC}. Outcomes are not always been on the positive 
side, as the undecidability of the model checking problem for probabilistic 
variations on HORSs shows. Again, we are not aware of any study aimed at giving 
general criteria for decidability. 

\paragraph{Effect Handlers}
Effect handlers are a linguistic feature allowing to give computational meaning 
to algebraic operations through handlers, i.e. routines specifically dedicated 
to the management of these effects, which in this way can be managed internally 
by the program itself rather than by the environment. Given 
their elegance and 
naturalness in generalizing standard language construction like 
the $\mathtt{try}~\mathtt{with}$
operator for exceptions, handlers has been largely studied both theoretically and concretely 
\cite{Plotkin2009:Handlers,KammarICFP2013:Handlers,Hillerstorm2017:Continuation,Hillerstrom2018:Shallow,Bauer2019:Algebraic,Sekiyama2020:PolymorphicEffects,Biernackietal2019:AbstractingEffects}.
 We are not aware of any study about handlers in a finitary 
setting, and even less about questions of decidability with 
regards to higher-order model checking.

\section{Conclusion}
This paper tackles, for the first time in a general way, the problem 
of evaluating the intrinsic difficulty of the higher-order model 
checking problem when applied to programs that exhibit effects, 
possibly managed through handlers. The results obtained are in two 
styles: while the problem of capturing algebraic operations in calculi 
amenable to HOMC does not pose particular problems and indeed can be 
solved in its generality, observing the effects produced by such 
operations and handling them must be done with great care: we observe 
that in general this leads to undecidability, but that in both cases, 
criteria can be defined allowing to keep the problem decidable. This 
consists, respectively, in observing the effects so that this 
observation can be expressed in \MSO\ and in a restricted class of 
handlers sufficient for the modeling of the so-called generic effects.

There are aspects that this paper deliberately overlooks, but which 
certainly deserve further study. First of all, it should be mentioned 
that the impact of effects on (known) complexity results about HOMC is 
not studied in detail here, but that the introduced translations (in 
particular those in Section~\ref{sect:translationEPCFtoY} and Section~\ref{s:modelcheckinghandlers}) could be implemented more efficiently, following works on order optimization \cite{NakamuraetAl2020:AverageCaseHardness}. 
It should also be said that the possibility of introducing logics more 
powerful than \MSO\ this way capturing quantitative observations without 
possibly falling back into the known cases of undecidability (e.g., in 
the case of probabilistic effects) is a very interesting direction 
that the authors intend to investigate in the immediate future.

\bibliographystyle{ACM-Reference-Format}
\bibliography{biblio.bib}

\end{document}